\documentclass[11pt]{article}
 
\usepackage{epsfig}
\usepackage[usenames]{color}
\usepackage{amsmath}
\usepackage{amssymb}
\usepackage{mathrsfs}
\usepackage{amsfonts}
\usepackage{lscape}
\usepackage{graphicx}
\usepackage{amsthm}
\usepackage{bbm}
\usepackage{todonotes}
\usepackage{hyperref}
\hypersetup{colorlinks=false,pdftitle="Entropy and set cardinality inequalities for partition-determined functions",pdftex}

\theoremstyle{plain} 

\theoremstyle{plain}
\newtheorem{lem}{Lemma}[section]
\newtheorem{thm}[lem]{Theorem}
\newtheorem{cor}[lem]{Corollary}
\newtheorem{example}[lem]{Example}
\newtheorem{conj}[lem]{Conjecture}

\theoremstyle{definition}

\newtheorem{rmk}{Remark}
\newtheorem{define}{Definition}
\newtheorem*{claim*}{Claim}

\newcommand{\R}{\mathbbm{R}}
\newcommand{\Z}{\mathbbm{Z}}
\newcommand{\scrG}{\mathcal{G}}

\newcommand{\bc}{\begin{center}}
\newcommand{\ec}{\end{center}}
\newcommand{\bt}{\begin{tabular}}
\newcommand{\et}{\end{tabular}} 
\newcommand{\bea}{\begin{eqnarray}}
\newcommand{\eea}{\end{eqnarray}}
\newcommand{\bean}{\begin{eqnarray*}}
\newcommand{\eean}{\end{eqnarray*}}

\newcommand{\ba}{\begin{array}}
\newcommand{\ea}{\end{array}}

\def\be{\begin{eqnarray}}
\def\ee{\end{eqnarray}}
\def\ben{\begin{eqnarray*}}
\def\een{\end{eqnarray*}}

\newcommand{\ra} {\rightarrow}

\newcommand{\leqa}{\mbox{$ \;\stackrel{(a)}{\leq}\; $}}

\newcommand{\geqb}{\mbox{$ \;\stackrel{(b)}{\geq}\; $}}

\newcommand{\eqb}{\mbox{$ \;\stackrel{(b)}{=}\; $}}
\newcommand{\eqc}{\mbox{$ \;\stackrel{(c)}{=}\; $}}

\newcommand{\RL}{{\mathbb R}}

\newcommand{\BBP}{{\mathbb P}}

\newcommand{\calM}{\mbox{${\cal M}$}}

\newcommand{\calX}{\mbox{${\cal X}$}}

\newcommand{\Rpl}{{\mathbb R}_{+}}

\def\elabel#1{\label{e:#1}}

\def\sq{$\Box$}

\def\qed{\ifmmode\sq\else{\unskip\nobreak\hfil
\penalty50\hskip1em\null\nobreak\hfil\sq
\parfillskip=0pt\finalhyphendemerits=0\endgraf}\fi\par\medbreak}

\newsavebox{\junk}
\savebox{\junk}[1.6mm]{\hbox{$|\!|\!|$}}

\def\til={{\widetilde =}}

 \def\eq#1/{(\ref{#1})}

\def\eq#1/{(\ref{e:#1})}

\newcommand{\beqn}[1]{\notes{#1}%
\begin{eqnarray} \elabel{#1}}

\newcommand{\eeqn}{\end{eqnarray} }

\newcommand{\beq}[1]{\notes{#1}%
\begin{equation}\elabel{#1}}

\newcommand{\eeq}{\end{equation}} 

\def\bdes{\begin{description}}
\def\edes{\end{description}}

\def\notes#1{}

\newcommand{\setS}{s}
\newcommand{\setT}{t}
\newcommand{\Xs}{{\bf X}_{\setS}}
\newcommand{\Xt}{{\bf X}_{\setT}}

\newcommand{\Zsp}{Z^{+}_{\setS}}
\newcommand{\collS}{\mathcal{C}}

\newcommand{\sumS}{\sum_{\setS\in\collS}} 
\newcommand{\prodS}{\prod_{\setS\in\collS}} 

\newcommand{\as}{\alpha_{\setS}}

\newcommand{\bs}{\beta_{\setS}}
\newcommand{\bZ}{Z} 
\newcommand{\Zs}{Z_{\setS}}

\newcommand{\calA}{\mathcal{A}}
\newcommand{\calB}{\mathcal{B}}
\newcommand{\calG}{\mathcal{G}}

\newcommand{\nullset}{\phi}
\newcommand{\bX}{X} 
\newcommand{\bXn}{{\bf X}_{[n]}}

\newcommand{\bXk}{{\bf X}_{[k]}}

\newcommand{\bx}{{\bf x}}
\newcommand{\by}{{\bf y}}

\newcommand{\naplus}{\circ}

\begin{document}

\title{Entropy and set cardinality inequalities for partition-determined functions} 
\author{
Mokshay Madiman\thanks{Department of Statistics, Yale University,
24 Hillhouse Avenue, New Haven, CT 06511, USA.
Email: \texttt{mokshay.madiman@yale.edu}}
\and 
Adam W. Marcus\thanks{Department of Mathematics,
Yale University, PO Box 208283, New Haven, CT 06520, USA.
Email: \texttt{adam.marcus@yale.edu}}
\and
Prasad Tetali\thanks{School of Mathematics and School of Computer Science,
Georgia Institute of Technology,
Atlanta, GA 30332-0160, USA.
Email: \texttt{tetali@math.gatech.edu}}
}
\date{}
\maketitle

\begin{abstract}
A new notion of partition-determined functions is introduced, and several basic inequalities
are developed for the entropy of such functions of independent random variables,
as well as for cardinalities of compound sets obtained using these functions. Here a compound set means a set obtained by
varying each argument of a function of several variables over a set associated with that argument,
where all the sets are subsets of an appropriate algebraic structure so that the function is
well defined. On the one hand, the entropy inequalities developed for partition-determined functions 
imply entropic analogues of general inequalities of Pl\"unnecke-Ruzsa type. On the other 
hand, the cardinality inequalities developed for compound sets imply several inequalities for sumsets, 
including for instance a generalization of inequalities proved by Gyarmati, Matolcsi and Ruzsa (2010).
We also provide partial progress towards a conjecture of Ruzsa (2007) for sumsets in nonabelian groups.
All proofs are elementary and rely on properly developing certain information-theoretic inequalities.
\end{abstract}


\section{Introduction}
\label{sec:intro}

It is well known in certain circles that there appears to exist an informal
parallelism between entropy inequalities on the one hand, 
and set cardinality inequalities on the other. In this paper, we clarify some aspects
of this parallelism, while presenting new inequalities for both
entropy and set cardinalities.

A natural connection between entropy and set cardinalities arises from
the fact that the entropy of the uniform distribution on a finite set of
size $m$ is just $\log m$, and this is the maximum entropy of any
distribution supported on the same set. Consequently,
inequalities for entropy lead to inequalities for
cardinalities of sets. For instance, by choosing
$(X,Y)$ to be uniformly distributed on the set $A\subset B\times C$,
the classical inequality
$H(X,Y)\leq H(X) +H(Y)$ implies
$\log |A| \leq \log |B| +\log |C|$ or $|A|\leq |B|\cdot |C|$. 


For the joint entropy, there is an elaborate history of 
entropy inequalities starting with the chain rule of Shannon,
whose major developments include works of Han, Shearer,
Fujishige, Yeung, Mat\'u$\check{\text{s}}$, and others. The classical part of this
work, involving so-called Shannon-type inequalities that use
the submodularity of the joint entropy, was synthesized and generalized
in \cite{MT07:isit, MT10}, where an array of
lower as well as upper bounds were given, for the joint entropy of a collection
of random variables generalizing inequalities of Han \cite{Han78}, 
Fujishige \cite{Fuj78} and Shearer \cite{CGFS86}. 
For the history of the non-classical part, involving so-called non-Shannon inequalities, 
one may consult for instance, Mat\'u$\check{\text{s}}$ \cite{Mat07}
and references therein.


Entropies of sums, even in the setting of independent summands, are 
not as well understood as joint entropies.
For continuous random variables, the so-called entropy power inequalities provide
important lower bounds on entropy of sums, see, e.g., \cite{MB07}.
For discrete random variables, an unpublished work of Tao and Vu \cite{TV06:unpub}
gives some upper bounds on entropy of sums, and discusses the analogy between
entropy of sums of random variables (instead of joint entropy)
and sumset cardinality inequalities (instead of set projection inequalities). 
In this paper, we develop this analogy in directions not
considered in \cite{TV06:unpub}, and prove more general inequalities
both for entropies and sumsets.  

The property of sumsets that allows us to apply ideas from entropy is that, 
for a fixed element $a$, the sum $a+b$ depends only on $b$ 
(no further knowledge of how $a$ and $b$ are related is needed).
We formalize this idea into what we will call a ``partition-determined'' function.

Let $X_1, X_2, \dots, X_k$ be finite sets.
Any nonempty subset $\setS \subset [k]$ consisting of the elements  $i_1<i_2< \ldots< i_{|\setS|}$
corresponds to a different product space 
$\Xs = \prod_{i \in \setS} X_i=\{ (x_{i_1}, \dots, x_{i_{|\setS|}}) : x_i \in X_i \}$.
For sets $\setS \subseteq \setT \subseteq [k]$, we define the projection function 
$\pi_{\setS}:\Xt \to \Xs$ in the natural way: 
for $x \in \Xt$, let $\pi_{\setS}(x) = (x_{i_1}, \dots, x_{i_{|\setS|}})$ where $i_j \in \setS$.
(To avoid ambiguity in the definition of $\Xs$ and $\pi_{\setS}$, we always assume 
that the indices $i_j$ are labeled in increasing order, i.e., $i_{j+1}>i_{j}$.)
When the meaning is clear, we will write $\pi_i(x)$ for $\pi_{\{i\}}(x)$.

We will use $Q(X_1, X_2, \dots, X_k)$ to denote the space that is a disjoint union 
of the spaces $\Xs$ corresponding to nonempty $\setS \subseteq [k]$, i.e.,
\ben 
Q(X_1, X_2, \dots, X_k) = \bigsqcup_{\nullset\neq\setS \subset [n]} \Xs .
\een
Let $Y$ be any space and $f:Q(X_1, \dots, X_k) \to Y$ be any function.
Then, for a nonempty set $\setS \subset [k]$, we define $f_{\setS}:\Xs \to Y$ to be the restriction 
of $f$ to only those inputs that came from $\Xs$.
We will abuse notation by writing, for nonempty $\setS \subseteq \setT$ and 
$x \in \setT$, $f_{\setS}(x)$ to mean $f_{\setS}(\pi_{\setS}(x))$, and when the domain is clear, we will merely write $f(x)$.

Let $\setS$ be a subset of $[k]$ and let $\overline{\setS}$ denote $[k] \setminus \setS$.
We will say that a function $f$ defined on $Q(X_1, X_2, \dots, X_k)$ is 
{\em partition-determined with respect to} $\setS$ if for all $x,y \in X_{[k]}$, 
we have that $f(x) = f(y)$ whenever both $f_{\setS}(x) = f_{\setS}(y)$ and 
$f_{\overline{\setS}}(x) = f_{\overline{\setS}}(y)$ 
(informally, $f_{\setS}(x)$ and $f_{\overline{\setS}}(x)$ uniquely determine the value of $f(x)$).
Extending this idea to collections of subsets $\collS$, we will say that $f$ is 
{\em partition-determined with respect to} $\collS$ if $f$ is partition-determined 
with respect to $\setS$ for all $\setS \in \collS$.
Finally, in the case that $f$ is partition-determined with respects to all subsets $[k]$, 
we will simply write that $f$ is {\em partition-determined}.

An alternate and more direct construction of the above definitions was suggested by I.~Ruzsa \cite{Ruz:comm}
after the first version of this paper was circulated.
Specifically, extend each set $X_i$ by an element called $\infty$, i.e., set $X_i'=X_i\cup \{\infty\}$.
Then the set  $Q'(X_1, \dots, X_k)=\prod_{i \in [k]} X_i'$ has one more element than $Q(X_1, X_2, \dots, X_k)$,
and consists only of $k$-tuples. One would have to define the projection functions  $\pi_{\setS}$ differently:
$\pi_{\setS}(x)$ would set to $\infty$ all the components of $x$ indexed by $\overline{\setS}$, and leave the others
unchanged. Then one can say that $f$ defined on $Q'(X_1, X_2, \dots, X_k)$ is 
partition-determined with respect to $\setS$ if for all $x,y \in {\bf X}_k$, 
we have that $f(x) = f(y)$ whenever both $f_{\setS}(x) = f_{\setS}(y)$ and $f_{\overline{\setS}}(x) = f_{\overline{\setS}}(y)$.
In the sequel, we stick with our original notation.

The definition above is intended to capture the property of sumsets that was mentioned earlier.
Simple examples relevant for consideration include Cartesian products of sets and linear combinations of sets 
(and so, in particular, sumsets).  Both of these classes of examples
are partition-determined with respect to $\collS$ for any $\collS$.

\begin{example}\label{entropy:product}
Let $V$ be a vector space over the reals with basis vectors $\{v_1, \dots, v_k \}$.  
Let $X_1, \dots, X_k \subseteq \R$ and define $f:Q(X_1, \dots, X_k) \to V$ such that $f_{\setS}(x) = \sum_{i \in \setS} \pi_i(x) v_i$.
Then $f$ is partition-determined with respect to $\collS$ for all collections $\collS$ of subsets of $[k]$.
\end{example}
\begin{proof}
Let $x \in X_{\setT}$ for some $\setT \subseteq [k]$ and let $\setS \in \collS$ where $\collS$ is a collection of subsets of $[k]$.
Then 
\[ f(x) ~=~ \sum_{i \in \setT} \pi_i(x)v_i 
~= \sum_{i \in (\setS \cap \setT)} \pi_i(x)v_i ~+ \sum_{i \in (\overline{\setS} \cap \setT)} \pi_i(x)v_i 
~=~ f_{\setS}(x) + f_{\overline{\setS}}(x). \]
Thus knowing $f_\setS(x)$ and $f_{\overline{\setS}}(x)$ uniquely determines $f(x)$.
Since this is true for any $\setS \in \collS$, $f$ is partition-determined with respect to $\collS$.
\end{proof}

\begin{example}\label{entropy:linear}
Let $(\calG,+)$ be an abelian group and $X_1, \dots, X_k \subseteq \calG$ and let $c_1, \dots, c_k \in \Z$.
Define $f:Q(X_1, \dots, X_k) \to \calG$ such that $f_{\setS}(x) = \sum_{i \in \setS} c_i \pi_i(x)$.
Then $f$ is partition-determined with respect to $\collS$ for all collections $\collS$ of subsets of $[k]$.
\end{example}
\begin{proof}
The proof is identical to Example~\ref{entropy:product}, only replacing $v_i$ with $c_i$.
\end{proof}


Equipped with the notion of partition-determined functions, we prove an
array of inequalities for both entropy and set cardinality. For instance, we have
the following results for sums as corollaries of general statements for 
partition-determined functions.

\par\vspace{.1in}
\noindent{\bf Illustrative Entropy Result:}
Let $Z_{1},\dots,Z_{n}$ be independent discrete random variables taking values in
the abelian group $(\calG, +)$, and let $\collS$ be an $r$--regular hypergraph on $[n]$. Then
\ben
H(Z_{1}+\dots +Z_{n}) \leq 
\frac{1}{r} \sumS  H\bigg( \sum_{i\in\setS} Z_{i} \bigg) .
\een

\par\vspace{.1in}
\noindent{\bf Illustrative Set Cardinality Result:}
Let $A, B_1, B_2, \dots, B_n \subset \calG$ be finite subsets of the abelian group $(\calG, +)$. 
If $\collS$ is an $r$-regular hypergraph on $[n]$, then for any $D \subseteq B_1+\ldots+B_n$,
\ben
|A+D|^{|\collS|} \leq |D|^{{|\collS|}-r} \prodS \bigg|A+  \sum_{i\in\setS} B_i \bigg| .
\een
\par\vspace{.1in}

This set inequality (and others for sums, projections etc.) is obtained as a corollary
of inequalities for cardinalities of more general structures, which we call {\em compound sets}.

\begin{define}
A {\em compound set} is a set of the form 
$\{f(x_1,\ldots,x_k): x_1\in X_1,\ldots, x_k\in X_k  \}$,
where the sets $X_1,\ldots,X_k$  are subsets of an appropriate algebraic structure $\calX$ so that $f$ is well defined.
This compound set is denoted $f(X_1,\ldots,X_k)$.
\end{define}

 
This paper is organized as follows. Section~\ref{ss:mi} presents preliminaries
on mutual information, and a key lemma on entropies of partition-determined
functions. Section~\ref{ss:ent-sm} presents a rather general new submodularity property 
of the entropy for strongly partition-determined functions (defined there)
of independent random variables. Surprisingly, this result is entirely elementary and 
relies on the classical properties of joint entropy. 
Section~\ref{ss:ent-UB} uses this result to
demonstrate new, general upper bounds on entropies of partition-determined functions.
Section~\ref{ss:ent-special} considers applications to particular partition-determined functions 
such as sums of random variables taking values in abelian groups. In the latter setting,
we present entropic analogues of general versions of the Pl\"unnecke-Ruzsa inequalities
for multiple, possibly distinct summands.

Section~\ref{ss:set-basic} applies joint entropy inequalities
to obtain a basic but powerful result about compound sets. 
The remaining subsections of Section~\ref{sec:set} 
discuss various consequences of the basic result for compound
set cardinalities in Section~\ref{ss:set-basic}. 
Section~\ref{ss:set-proj} gives a first easy application of the compound set
result to obtaining well known cardinality inequalities for projections of sets. 

Section~\ref{ss:set-abelian} gives a second, and important, application
to sumsets in abelian groups. In this, we build on recent work by Gyarmati, Matolcsi 
and Ruzsa \cite{GMR10}, who noted that Han-type  inequalities 
can be applied to sumsets in much the same way that they can be applied to 
characteristic functions of sets of random variables (the usual situation).
While it is not true in general that sumsets satisfy a log-submodular relation in 
an obvious way, it is natural to ask whether they permit a weaker property, by way 
of fractional subadditivity. It is classical (and recently reviewed in 
\cite{MT10}) that fractional subadditivity is 
weaker than log-submodularity and more general than Han's inequalities.
Here, by extending an idea embedded in \cite{GMR10}, 
and making further use of entropy, we show a 
general fractional subadditivity property for sumsets that implies some of the 
results and conjectures in \cite{GMR10} as easy corollaries. In particular,
we obtain general upper bounds on the cardinality of sumsets involving
multiple, possible distinct summands, and comment on the connection with
generalized Pl\"unnecke-Ruzsa inequalities.

Section~\ref{ss:set-nonabelian} applies the compound set inequalities
to obtain results for sumsets in non-abelian groups, motivated by a 
conjecture of Ruzsa \cite{Ruz07}. In particular, we make partial progress 
towards resolving Ruzsa's conjecture.
In Section~\ref{ss:set-nonsum}, we present a novel application
of our basic result to obtaining cardinality bounds for compound sets in rings that are
more general than sumsets (such as the compound sets $g(A,B)$ in a ring obtained
from $g(a,b)=a^2+b^2$ or $g(a,b)=a^2-b^2$). 

It should be noted that multiple papers on this topic have appeared in the literature recently,
each of which provides unique and independent contributions to the overall goal 
of developing an entropy-based sumset calculus. 
Gyarmati, Matolcsi, and Ruzsa \cite{GMR10} prove a special case 
of our Corollary~\ref{cor:frac-sum-ineq}, along with complementary inequalities of 
Cauchy-Davenport type (which we do not discuss). Balister and Bollob\'as \cite{BB07} build on the 
ideas in both \cite{GMR10} and \cite{MT10} to develop a 
hierarchy of entropy and sumset inequalities-- in particular, they prove special cases 
of our Corollaries~\ref{cor:frac-sum-ineq} and \ref{cor:abelian}.  In a different direction, 
Ruzsa \cite{Ruz09} explores more deeply the relationship between entropy and 
cardinality inequalities for sums, and deduces (among other interesting results) 
a special case of the third part of our Corollary~\ref{cor:ent-sum} using Han's inequality. 
A second paper of  Gyarmati, Matolcsi, and Ruzsa \cite{GMR08} also 
independently proves a special case of our Corollary~\ref{cor:abelian}, along with 
some interesting generalizations of the Pl\"unnecke-Ruzsa inequality 
for different summands (whose entropy analogue we provide in Section~\ref{sec:ent}). 
We will point out in more detail the intersections with \cite{BB07, Ruz09, GMR10, GMR08} 
at the appropriate places in the individual sections.

In addition, after the first version of this paper was released, Tao \cite{Tao10} 
developed some entropic analogues of other sumset inequalities such as 
the Balog-Szemeredi-Gowers theorem and Freiman's theorem. Tao also has some pertinent discussion on 
his blog \cite{Tao09:blog}, where he observes that part of our Corollary~\ref{cor:ent-sum}
(the submodularity of entropy of sums in discrete groups)
is actually implicit in an old paper of Ka{\u\i}manovich and Vershik \cite{KV83},
and one of his comments there also suggests that the Ka{\u\i}manovich-Vershik result
may be useful in simplifying the proofs of some of the results of \cite{Tao10}. 

We wish to emphasize, however, that apart from the fact that most of our 
results are new in the generality stated even for sums, the approach we 
adopt using partition-determined functions provides a general and powerful 
framework for studying such inequalities, and in that sense goes beyond the 
aforementioned papers that prove similar inequalities for the specific context 
of sumsets. Furthermore, we would suggest that the development in parallel of 
entropy and cardinality inequalities for sums (following Ruzsa \cite{Ruz09}), 
and in particular an examination of the similarities and differences 
(as for example at the end of Section~\ref{ss:set-abelian}) sheds light on why certain inequalities are 
true and certain stronger forms are not, and thus provides greater intuition for the 
user of the sumset calculus.

\section{Entropy inequalities}
\label{sec:ent}

\subsection{Partition-determined functions and mutual information}
\label{ss:mi}

As usual, we denote by $[n]$ the index set $\{1,2,\dots,n\}$. 
Let  $Z_{[n]}=(Z_1, Z_2, \dots,Z_{n})$ be a collection of
random variables, and assume each $Z_i$ takes values in some finite set $X_i$.
If the probability distribution of $Z_{[n]}$ has
joint probability mass function $p(z_{[n]})=\BBP(Z_{[n]}=z_{[n]})$,
where $z_{[n]}=(z_1,\dots,z_n)$, then
the {\it entropy} of $Z_{[n]}$ is defined by
\ben
H(Z_{[n]})=\sum_{z_{[n]}\in X_1\times \ldots\times X_n} - p(z_{[n]})\log p(z_{[n]}) .
\een
Recall that the {\it conditional entropy} of $Z$ given $Y$,
denoted $H(Z \mid Y)$, is defined by taking the mean using the
distribution of $Y$ of the entropy of the conditional distribution
of $Z$ given $Y=y$. 
The standard fact that $H(Z,Y)=H(Z)+H(Y|Z)$
has come to be known as Shannon's chain rule for entropy. For any
function $f$, it is easy to see that $H(f(Z))\leq H(Z)$,
with equality if and only if $f$ is a bijection.

The {\it mutual information} between two jointly distributed random variables
$Z$ and $Y$ is defined by 
\ben
I(Z;Y)=H(Z)-H(Z \mid Y) ;
\een
this is always nonnegative, and is a measure of the dependence between $Z$ and $Y$.
In particular, $I(Z;Y)=0$ if and only if $Z$ and $Y$ are independent.

Analogous to conditional entropy, one may also define the {\em conditional mutual information} 
$I(Z;Y \mid X)$, which quantifies how much more information 
about $Y$ one can glean from the pair $(X,Z)$ as compared to simply $X$.
More precisely, we can define 
\begin{align}
\label{indep}
I(Z;Y \mid X)&= I((X,Z);Y) -I(X;Y) \\
&= H(X,Z)-H(X,Z \mid Y) -[H(X)-H(X \mid Y)]\\
&= H(X,Z)-H(X,Z,Y) -H(X)+H(X,Y)  ,
\end{align}
where the alternate form in the last display was obtained by adding and subtracting $H(Y)$ in the last step, and 
using Shannon's chain rule for entropy.
The following lemma gives a simple and classical property of mutual information.

\begin{lem}
The mutual information cannot increase when one looks at functions
of the random variables (the  ``data processing inequality''):
\ben
I(f(Z);Y) \leq I(Z;Y) . 
\een
\end{lem}

A proof can be found in elementary texts on information theory such as 
Cover and Thomas \cite{CT91:book}.
The following strengthened notion of partition-determined functions will turn out to be useful:
\begin{define}
We say that $f:Q(X_1, \dots, X_n) \to Y$ is {\em strongly partition-determined} if for any disjoint sets
$\setS$ and $\setT$, the values of $f_{\setS\cup\setT}$ and $f_{\setT}$ (together) completely determine the value of $f_{\setS}$.
\end{define}

One can observe that both running examples, namely projections and sums,
are strongly partition-determined functions.
For brevity, we simply write $H(f_{\setS})$ for $H(f_{\setS}(Z))$.

\begin{lem}\label{lem:mi}
Suppose $X_i$ are finite sets, and 
$f:Q(X_1, \dots, X_n)\ra V$ is a partition-determined function.
Let $Z_1,\ldots,Z_n$ be random variables,
with $Z_i$ taking values in $X_i$.
Then, for disjoint sets $\setS,\setT\subset [n]$,
\be
I(f_{\setS\cup\setT};f_{\setT})\geq H(f_{\setS\cup\setT})-H(f_{\setS}).
\ee
If, furthermore, $f$ is strongly partition-determined and
$Z_1,\ldots,Z_n$ are independent, then 
\be\label{chain1}
I(f_{\setS\cup\setT};f_{\setT})=H(f_{\setS\cup\setT})-H(f_{\setS}).
\ee
\end{lem}
\begin{proof}
Since conditioning reduces entropy,
\ben
H(f_{\setS\cup\setT})-H(f_{\setS}) \leqa H(f_{\setS\cup\setT})-H(f_{\setS}|f_{\setT}).
\een
But since $f$ is partition-determined, $f_{\setS\cup\setT}=\phi(f_{\setS},f_{\setT})$ for some
function $\phi$, and hence
\ben
H(f_{\setS}|f_{\setT})= H(f_{\setS},f_{\setT}|f_{\setT}) \geqb H(f_{\setS\cup\setT}|f_{\setT}) .
\een
Thus
\ben
H(f_{\setS\cup\setT})-H(f_{\setS}) \leq H(f_{\setS\cup\setT})-H(f_{\setS\cup\setT}|f_{\setT})
=I(f_{\setS\cup\setT};f_{\setT}).
\een
This yields the first part of the lemma. For the second part,
note that independence of $Z_1,\ldots,Z_n$ guarantees equality in (a),
while $f$ being strongly partition-determined guarantees equality in (b).
\end{proof}

\subsection{A basic result for entropy}
\label{ss:ent-sm}

The inequality below, while a simple consequence of the above elementary facts, is rather powerful.

\begin{thm}\label{thm:entSM}{\sc [Submodularity for Strongly Partition-determined Functions]}\, 
Suppose $X_i$ are finite sets, and $f:Q(X_1, \dots, X_n)\ra V$ is a strongly partition-determined function.
Let $Z_1,\ldots,Z_n$ be independent random variables, with $Z_i$ taking values in $X_i$.
Then
\be
H(f_{\setS\cup\setT}) + H(f_{\setS\cap\setT})
\leq H(f_{\setS}) + H(f_{\setT})
\ee
for any nonempty subsets $\setS$ and $\setT$ of $[n]$.
\end{thm}
\vspace{-.2in}
\hspace{-.4in}\begin{proof}
First note that it suffices to prove the result for $n=3$ (since we can consider collections of random variables to be a single random variables under the joint distribution).
For this case, we have
\begin{align*}
&H(f_{\{1,2\}})+H(f_{\{2,3\}})
- H(f_{\{1,2,3\}}) - H(f_{\{2\}}) \\
&= H(f_{\{1,2\}})-H(f_{\{2\}}) - 
\big[ H(f_{\{1,2,3\}}) - H(f_{\{2,3\}}) \big]\\
&= I(f_{\{1,2\}};f_{\{1\}}) - I(f_{\{1,2,3\}};f_{\{1\}}),
\end{align*}
using \eqref{chain1} from Lemma~\ref{lem:mi}.
Thus we simply need to show that
\ben
I(f_{\{1,2\}};f_{\{1\}}) \geq I(f_{\{1,2,3\}};f_{\{1\}}) .
\een
Now
\begin{align*}
 I(f_{\{1,2,3\}};f_{\{1\}}) 
&\leqa  I(f_{\{1,2\}}, f_{\{3\}} \, ; f_{\{1\}}) \\
&\eqb  I(f_{\{1,2\}} ;f_{\{1\}}) + I(f_{\{3\}} ;f_{\{1\}} | f_{\{1,2\}} ) \\
&\eqc  I(f_{\{1,2\}} ;f_{\{1\}}) 
\end{align*}
where (a) follows from the data processing inequality, (b)
follows from \eqref{indep}, and (c) follows from the hypothesis of independence;
so the proof is complete.
\end{proof}

\begin{rmk}\label{rmk:null}
Observe that we can allow the consideration of empty sets in Theorem~\ref{thm:entSM}
(and below) if we set $H(f_\nullset)=0$. Indeed, this is natural because $f_\nullset(x)=f(\pi_\nullset(x))$
does not have any actual arguments by definition, and hence must be the constant function,
for which the entropy is of course 0.
\end{rmk}

A consequence of the submodularity of 
entropy for strongly partition-determined functions 
is an entropy inequality that obeys 
the partial order constructed using compressions,
as introduced by Bollob\'as and Leader \cite{BL91}.
Following Balister and Bollob\'as \cite{BB07}, we introduce 
some notation. 
Let $\calM(n,m)$ be the following family of 
(multi)hypergraphs: each consists of non-empty (ordinary) subsets $\setS_i$ of $[n]$,
$\setS_i=\setS_j$ is allowed, and 
$\sum_i |\setS_i|= m$. 
Consider a given multiset 
$\collS = \{\setS_1, \dots, \setS_l\} \in \calM(n,m)$.
The idea is to consider an operation that takes two sets in $\collS$
and replaces them by their union and intersection; however, note that
(i) if $\setS_i$ and $\setS_j$ are nested (i.e., either $\setS_i \subset \setS_j$ or 
$\setS_j \subset \setS_i$), then replacing $(\setS_i,\setS_j)$ by 
$(\setS_i \cap \setS_j ,\setS_i \cup \setS_j)$ does not change $\collS$, 
and
(ii) if  $\setS_i \cap \setS_j = \nullset$, the null set may enter the collection,
which would be undesirable.
Thus, take any pair of non-nested sets $\{\setS_i,\setS_j\}\subset \collS$
and let $\collS' = \collS(ij)$ be obtained from $\collS$ by replacing $\setS_i$ and $\setS_j$ 
by $\setS_i \cap \setS_j$ and $\setS_i \cup \setS_j$, 
keeping only $\setS_i \cup \setS_j$ if $\setS_i \cap \setS_j = \nullset$. 
$\collS'$ is called an {\it elementary compression} of $\collS$. The result of a sequence of
elementary compressions is called a {\it compression}.

Define a partial order on $\calM(n,m)$ by setting $\calA > \calB$ if $\calB$ is a
compression of $\calA$. To check that this is indeed a partial order, 
one needs to rule out the possibility of cycles, which can be done by noting
that
\ben
\sum_{\setS\in\calA} |\setS|^2 < \sum_{\setS\in\calA'} |\setS|^2 ,
\een
if $\calA'$ is an elementary compression of $\calA$.

\begin{thm}\label{thm:compr}
Suppose $X_i$ are finite sets, and $f:Q(X_1, \dots, X_n)\ra V$ is a strongly partition-determined function.
Let $Z_1,\ldots,Z_n$ be independent random variables, with $Z_i$ taking values in $X_i$.
Let $\calA$ and $\calB$ be finite multisets of subsets of $[n]$, with $\calA > \calB$.
Writing $f_{\setS}=f_{\setS}(Z_1,\ldots,Z_n)$,
\ben
\sum_{\setS\in\calA} H(f_{\setS}) \geq \sum_{\setT\in\calB} H(f_{\setT}) .
\een
\end{thm}

\begin{proof}
The proof follows exactly the same reasoning as given by Balister and Bollob\'as \cite{BB07}
for the special case of $f$ being the identity function (or $f_{\setS}$ being the projection function). 
When $\calB$ is an {\it elementary} compression of $\calA$, the statement is immediate 
from the submodularity of $H(f_{\setS})$ proved
in Theorem~\ref{thm:entSM}, and transitivity of the partial order gives the full statement. 
\end{proof}

Note that for every multiset $\calA \in \calM(n,m)$ there is a unique minimal multiset
$\calA^{\#}$ dominated by $\calA$ consisting of the sets
$\setS^{\#}_j = \{i \in [n] : i \text{ lies in at least } j \text{ of the sets } \setS \in \calA\}$.
Thus a particularly nice instance of Theorem~\ref{thm:compr} is for the special case of $\calB=\calA^{\#}$.

\subsection{Upper bounds for entropy of a partition-determined function}
\label{ss:ent-UB}

Let $\collS$ be a collection of subsets of $[n]$. 
For any index $i$ in $[n]$, define the {\it degree}
of $i$ in $\collS$ as $r(i)=|\{\setT\in\collS: i\in \setT\}|$.
A function $\alpha:\collS \ra \Rpl$,
is called a {\em fractional covering}, if for each $i\in [n]$, we have
$\sum_{\setS\in \collS:i\in \setS} \as \ge 1$.
If $\alpha$ satisfies the equalities $\sum_{\setS\in \collS:i\in \setS} \as = 1$
for each $i\in[n]$, it is called a {\it fractional partition}. If the degree of every index
$i$ in $\collS$ is exactly $r$, $\collS$ is called an $r$-regular hypergraph,
and $\as=1/r$ for every $\setS\in\collS$ constitutes a fractional partition using $\collS$.

\begin{thm}\label{thm:entUB}{\sc [Upper Bound for Entropy of Strongly Partition-determined Function]}\, 
Suppose $X_i$ are finite sets, and $f:Q(X_1, \dots, X_n)\ra V$ is a strongly partition-determined function.
Let $Z_1,\ldots,Z_n$ be independent random variables, with $Z_i$ taking values in $X_i$.
Then, writing $f_{\setS}=f_{\setS}(Z_1,\ldots,Z_n)$, we have
\ben
H(f_{[n]}) \leq 
\sumS \as H\big( f_{\setS}\big) ,
\een
for any fractional covering $\alpha$ using any collection $\collS$
of subsets of $[n]$.
\end{thm}

\begin{proof}
Define the set function $g(\setS)=H(f_{\setS})$, and note
that $g(\nullset)=0$ is the appropriate convention (see Remark~\ref{rmk:null}). 
Theorem~\ref{thm:entSM} says that $g$ is a submodular function.
Now the corollary follows from the general fact that a submodular
function $g$ with $g(\nullset)=0$ is ``fractionally subadditive'' (see, e.g., \cite{MT10}). 
Let us note that submodularity was also used by Han \cite{Han78} and 
by Chung, Graham, Frankl and Shearer \cite{CGFS86} for proving
Han's and Shearer's inequalities.
\end{proof}

For any collection $\collS$ of subsets, 
the degree covering (introduced in \cite{MT10}) is given by
\ben
\as=\frac{1}{r_{-}(\setS)},
\een
where $r_{-}(\setS)=\min_{i\in\setS} r(i)$. Specializing
Theorem~\ref{thm:entUB} to this particular fractional covering, we obtain
\ben
H(f_{[n]}) \leq 
\sumS \frac{1}{r_{-}(\setS)} H\big( f_{\setS}\big) .
\een
A simple example is the case of the collection $\collS_{m}$, 
consisting of all subsets of $[n]$ with $m$ elements, for which
the degree of each index with respect to $\collS_{m}$
is $\binom{n-1}{m-1}$.

\subsection{Corollaries for entropies of sums in abelian groups}
\label{ss:ent-special}

The first (and best studied) special case of interest is when $f$ is the identity
mapping, so that $f_{\setS}$  is the projection onto the subset $\setS$ of
coordinates. In this case, Theorems~\ref{thm:entSM}, \ref{thm:compr} and 
\ref{thm:entUB} reduce to the following fact. For any set $\setS\subset [n]$,
let $\Zs$ stand for the random variable $(Z_{i}:i\in \setS)$, with the indices taken
in their increasing order.
If $Z_{1},\dots,Z_{n}$ are independent discrete random variables taking values in
the group $\calG$, then
\ben
H(Z_{1},\ldots, Z_{n}) \leq 
\sumS \as H\big( \Zs\big) ,
\een
for any fractional covering $\alpha$ using any collection $\collS$
of subsets of $[n]$.
In the context of independent random variables, however, this fact is not particularly enlightening (although true, and with equality
for fractional partitions).
For more on its validity in the general dependent case, see \cite{MT10}.

A more interesting second case is when the sets $X_i$ are finite subsets
of an ambient abelian group $(\calG,+)$, and $f$ is just the sum function.

\begin{cor}\label{cor:ent-sum}{\sc [Entropy of Sums in abelian Groups]}\, 
Let $Z_{1},\dots,Z_{n}$ be independent discrete random variables taking values in
the abelian group $\calG$, and let 
\be\label{zsp}
\Zsp=\sum_{i\in\setS} Z_{i} .
\ee
Then:
\begin{enumerate}
\item The set function $f(\setS)=H(\Zsp)$ is submodular.
\item If $\calA > \calB$ (i.e., $\calB$ is a compression of $\calA$), then
\ben
\sum_{\setS\in\calA} H(\Zsp) \geq \sum_{\setT\in\calB} H(Z^{+}_{\setT}) .
\een
\item For any fractional covering $\alpha$ using any collection $\collS$
of subsets of $[n]$,
\ben
H(Z_{1}+\dots +Z_{n}) \leq 
\sumS \as H\big( \Zsp\big) .
\een
\end{enumerate}
\end{cor}

Since the sum function in an abelian group is strongly partition-determined,
the three statements in Corollary~\ref{cor:ent-sum} follow from
Theorems~\ref{thm:entSM}, \ref{thm:compr} and 
\ref{thm:entUB}  respectively. Parts of Corollary~\ref{cor:ent-sum} were
presented in \cite{Mad08:itw}. Let us note, in particular, that the first part of 
Corollary~\ref{cor:ent-sum} resolves affirmatively a strengthened form of
``Entropy Conjecture 3'' in the paper of Ruzsa \cite{Ruz09}. 
Indeed, that conjecture stated that for independent discrete random variables $Z_1, Z_2, Z_3$ taking
values in an abelian group, 
$H(Z_1)+H(Z_2+Z_3) \leq H(Z_1+Z_2) + H(Z_1+Z_3)$,
whereas Corollary~\ref{cor:ent-sum} contains this conjecture with the $H(Z_2+Z_3)$
term improved to $H(Z_1+Z_2+Z_3)$.

Furthermore, Corollary~\ref{cor:ent-sum} contains entropy analogues of the
Pl\"unnecke-Ruzsa inequalities \cite{Plu70,Ruz89, Ruz90}. In fact, we now deduce
from it entropy analogues of rather general  Pl\"unnecke-type 
inequalities for sumsets (see Section~\ref{ss:set-abelian}).

\begin{thm}\label{thm:ent-pl-r}{\sc [Entropic Analogue of Generalized Pl\"unnecke-Ruzsa inequality]}\, 
Let $Z_{0}, Z_{1},\dots,Z_{n}$ be {\em independent} discrete random variables taking values in
the abelian group $\calG$, and let $\Zsp$ be defined as in \eqref{zsp}.
Let $\alpha$ be any fractional covering  using the collection $\collS$
of subsets of $[n]$, and set  $c=\sumS \as$. Then the following statements are true:
\begin{enumerate}
\item 
\be\label{ent-pl-r1}
cH(Z_0+Z_{[n]}^{+})\leq (c-1) H(Z_{[n]}^{+}) + \sumS \as H(Z_0+\Zsp) .
\ee
\item
\be\label{ent-pl-r2}
H(Z_0+Z_{[n]}^{+})\leq \sumS \as H(Z_0+\Zsp) - (c-1) H(Z_0) .
\ee 
\end{enumerate}
\end{thm}

\begin{proof}
First set $Z_{n+1}=Z_0$ for convenience, and consider the following collection of subsets of $[n+1]$:
\ben 
\collS' = \left\{ [n] \right\} \cup \left\{ \setS': \setS'=\setS\cup\{n+1\} , \setS\in\collS \right\} .
\een
For each set $\setS\in\collS$, let 
$\gamma_{\setS\cup\{n+1\}} = \frac{\alpha_{\setS}}{c}$,
and let 
$\gamma_{[n]}= 1-\frac{1}{c}$ (this is clearly nonnegative since $c\geq 1$).
Observe that $\gamma$ is a fractional covering for $[n+1]$ using $\collS'$; for the index $n+1$, one has
\ben
\sumS \gamma_{\setS\cup\{n+1\}} = 1
\een
and for each $j\in[n]$, one has 
\ben
\gamma_{[n]} + \sum_{\setS\in\collS:\setS\ni j} \gamma_{\setS\cup\{n+1\}} 
=\bigg[1-\frac{1}{c}\bigg]  + \sum_{\setS\in\collS:\setS\ni j} \frac{\alpha_{\setS}}{c} \geq 1 ,
\een
where the last inequality follows from the assumption that $\alpha$ is a fractional covering.
The inequality \eqref{ent-pl-r1} follows by applying the third part of Corollary~\ref{cor:ent-sum} 
to this fractional covering.

The inequality \eqref{ent-pl-r2} is essentially contained in \cite{Mad08:itw}, but we give the proof for
completeness. First note that since the sum is a strongly
partition-determined function, Lemma~\ref{lem:mi} implies that
\ben
H(Z_{1}+Z_{2})=H(Z_{1})+I(Z_{1}+Z_{2};Z_{2})
\een
for independent random variables $Z_1$ and $Z_2$.
Repeatedly applying this gives the ``chain rule''
\be\label{eq:chain}\begin{split}
H\bigg( \sum_{i\in [n]} Z_{i}\bigg) 
= H(Z_1)+\sum_{i=2}^{n} I(Z_{[i]}^{+};Z_{i}) ,
\end{split}\ee
for independent random variables $Z_i$. In particular, we have in the current setting for each $\setS\in\collS$ that 
\ben
H(Z_0+\Zsp) =H(Z_0)+ \sum_{i\in\setS} I(Z_{i};Z_{\setS\cap [i]}^{+}) .
\een
Thus
\ben\begin{split}
\sumS \as H(Z_0+\Zsp) 
&= \sumS \as [ H(Z_0)+\sum_{i\in\setS} I(Z_{i};Z_{\setS\cap [i]}^{+}) ] \\
&\geq \sumS \as H(Z_0) + \sumS \as \sum_{i\in\setS} I(Z_{i};Z_{[i]}^{+}) ,
\end{split}\een
where we used the fact that $I(X;X+Y) \geq I(X;X+Y+Z)$ for independent random variables $X, Y, Z$
(this may either be seen in the proof of Theorem~\ref{thm:entSM}
specialized to sums, or by rewriting the first part of Corollary~\ref{cor:ent-sum}
using Lemma~\ref{lem:mi}).
By an interchange of sums, the definition of a fractional covering, and another use of the just-mentioned fact,
\ben\begin{split}
\sumS \as H(Z_0+\Zsp) 
&\geq \sumS \as H(Z_0)+  \sum_{i\in[n]} I(Z_{i};Z_{[i]}^{+}) \sum_{\setS\in\collS,\setS\ni i} \as \\
&\geq \sumS \as H(Z_0)+ \sum_{i\in[n]} I(Z_{i};Z_{[i]}^{+}) \\
&\geq \sumS \as H(Z_0)+ \sum_{i\in[n]} I(Z_{i}; Z_{0}+Z_{[i]}^{+}) .
\end{split}\een
Again using the chain rule \eqref{eq:chain} for sums gives
\ben\begin{split}
\sumS \as H(Z_0+\Zsp) \geq \bigg(\sumS\as -1\bigg) H(Z_0) + H(Z_{0}+Z_{[n]}^{+}) ,
\end{split}\een
which on rearrangement is the desired inequality.
\end{proof}

To write Theorem~\ref{thm:ent-pl-r}  in a form reminiscent of the usual way of writing Pl\"unnecke-Ruzsa inequalities,
define the nonnegative constants $\bs$ for each $\setS\in\collS$ by
\ben
H(Z_0+\Zsp)= H(Z_0)+\bs  .
\een
Then Theorem~\ref{thm:ent-pl-r} says that
\be\label{ent-pl-r1b}
H(Z_0+Z_{[n]}^{+})\leq H(Z_0)+\bigg(1-\frac{1}{c}\bigg) H(Z_{[n]}^{+}) +  \frac{1}{c}\sumS \as \bs ,
\ee
and
\be\label{ent-pl-r2b}
H(Z_0+Z_{[n]}^{+})\leq H(Z_0)+\sumS \as \bs .
\ee
Observe that neither of these bounds seems {\it a priori} better than the other. In fact,
comparing the two bounds boils down to comparing $H(Z_{[n]}^{+})$ and
$\sumS \as \bs$, and we can write
\ben
\sumS \as \bs = \sumS \as [H(Z_0+\Zsp)- H(Z_0)]
\leq \sumS \as H(\Zsp) ;
\een
this last also being an upper bound for $H(Z_{[n]}^{+})$ by the third part of Corollary~\ref{cor:ent-sum}.
Thus both \eqref{ent-pl-r1b} and \eqref{ent-pl-r2b} have a common weakened form, namely
\ben
H(Z_0+Z_{[n]}^{+})\leq H(Z_0)+\bigg(1-\frac{1}{c}\bigg) \sumS \as H(\Zsp) +  \frac{1}{c}\sumS \as [H(Z_0+\Zsp)- H(Z_0)]  .
\een

Later we will prove Theorem~\ref{thm:abelian}, which is precisely the set cardinality analogue
of the inequality \eqref{ent-pl-r1b}. On the other hand, a set cardinality analogue of the
inequality \eqref{ent-pl-r2b} for the collections $\collS_m$ consisting of all sets of size $m$ was
recently proved by Gyarmati, Matolcsi and Ruzsa \cite[Theorem 1.3]{GMR08},
generalizing the original Pl\"unnecke-Ruzsa inequality for different summands. 
See Section~\ref{ss:set-abelian} for elaboration and further discussion.

Let us specialize Theorem~\ref{thm:ent-pl-r} to the collection $\collS_1$ of singleton sets.
Fixing  
$H(Z_0+Z_i)= H(Z_0)+\beta_i$ 
for each $i=1,\ldots,n$ gives
\ben 
H(Z_0+ Z_{[n]}^{+} )\leq H(Z_0)+ (1-n^{-1}) H(Z_{[n]}^{+}) + n^{-1}\sum_{i\in [n]} \beta_i ,
\een
and
\ben
H(Z_0+ Z_{[n]}^{+} )\leq H(Z_0)+ \sum_{i\in [n]} \beta_i .
\een
This last inequality is the entropic analogue of the 
original Pl\"unnecke-Ruzsa inequality for different summands \cite{Ruz89},
as independently observed by Tao \cite{Tao09:blog}. 

\begin{rmk}
A natural question here is whether the original Pl\"unnecke inequality also has an entropy
analogue. Recall that the  original Pl\"unnecke inequality stated that assuming 
$|A+lB|\leq \alpha^{l} |A|$, then for every $k>l$,
\be\label{orig-plu}
|A+kB|\leq \alpha^{k} |A| ,
\ee
where $kB$ refers to the sumset $B+\ldots+B$ with $k$ summands. 
(The original Pl\"unnecke-Ruzsa inequality for different summands was
a generalization of the $l=1$ case of \eqref{orig-plu}.) The entropic analogue of 
the original Pl\"unnecke inequality is the statement that if $X, Y_1, Y_2, \ldots$ are
independent discrete random variables taking values in an abelian group with all the $Y_i$
identically distributed, then
$H(X+Y_1+\ldots+Y_l) \leq H(X) + l\alpha$
implies
\be\label{ent-orig-plu}
H(X+Y_1+\ldots+Y_k) \leq H(X) + k\alpha
\ee
for any $k>l$. That this is true follows by considering the following multihypergraph:
writing $X=Y_0$, let
\ben
\collS= \big\{ \setS\subset\{0,1,\ldots k\}: \setS=\{0\}\cup\{i,i+1,\ldots, k, 1, \ldots, l-k+i-1\}  \big\} .
\een
(That is, arrange $1,\ldots, k$ on a circle, and consider all contiguous sets of $l$ elements along
the circle-- along with 0-- as the sets in $\collS$). Observe that the minimal multihypergraph
$\collS^{\#}$ dominated by $\collS$ is given, as per the discussion at the end of Section~\ref{ss:ent-sm},
by the following collection of sets: 
$\setS^{\#}_j =\{0,1,\ldots,k\}$ for $j=1,\ldots,l$, 
and 
$\setS^{\#}_j =\{0\}$ for $j=l+1,\ldots,k$. Applying Theorem~\ref{thm:compr}
immediately yields the inequality \eqref{ent-orig-plu}.
\end{rmk}

\begin{rmk}
While the simplest applications of Corollary~\ref{cor:ent-sum} 
are to discrete abelian groups, it continues to hold for arbitrary abelian 
groups $\calG$ provided the random variables 
$Z_i$ are supported on discrete subsets of $\calG$. 
In fact, as shown in \cite{Mad08:itw} and further developed in \cite{Mad11:laa, MK10:isit, BM10:repi},
a modified version of Corollary~\ref{cor:ent-sum} holds even beyond the discrete 
support setting, although the arguments have to be modified to
deal with continuous (or differential) entropy. 
Specifically, one can work with $\calG$-valued random variables on  
any abelian locally compact topological group $\calG$, whose distribution is absolutely continuous with
respect to the Haar measure on $\calG$. These continuous versions have applications to 
matrix analysis \cite{Mad11:laa}, as well as to convex geometry \cite{BM10:repi, BM11:cras}.
Here we simply mention that the inequality  \eqref{ent-pl-r2} continues to hold
for independent random variables in the continuous setting, where the entropy 
can be negative (see \cite{Mad08:itw}).
\end{rmk}


\section{Set cardinality inequalities}
\label{sec:set}

\subsection{A basic result for reduced sets}
\label{ss:set-basic}

The proofs in this section extend the arguments of Gyarmati, Matolcsi, and Ruzsa \cite{GMR10} 
(more precisely, an idea in Theorem~1.2 of their paper) while making further use of entropy.
Given a collection of spaces $X_i$, their idea was to endow each $X_i$ with an arbitrary linear order.
Then for any set $\setS \subseteq [k]$ and elements $A, B \in X_\setS$, they define $A <_{lex} B$ if the 
vector representation of $A$ comes before that of $B$ in lexicographical order.

In the statements and proofs that follow, we will work in the following setting:
Let $X_1, X_2, \dots, X_k$ be finite sets, and $f$ be a function defined on 
$Q(X_1, X_2, \dots, X_k)$. 
Now let $Y \subseteq f(\bXk)$ be a given set that we wish to bound in size.
For each $y \in Y$, we define $r(y)$ to be the smallest element of $f^{-1}(y) \subseteq \bXk$ in 
lexicographical order and set $R = \{r(y) : y \in Y\}$.
Thus each $y \in Y$ has a unique ``representative preimage" $r(y)$, and $|Y|=|R|$.

The main observation is the following somewhat surprising lemma, which requires nothing more than the elementary properties of entropy discussed above; 
this is also the essence (in addition to Han's Inequality) of the proof of Theorem~1.2 in \cite{GMR10}.

\begin{lem}\label{entropy:L}
Let ${\bZ}$ be a random variable that chooses uniformly from the elements of $R$.
If $f$ is partition-determined with respect to $\setS$, then
\ben 
H\big(\Zs \mid f(\Zs)\big) = 0.
\een
\end{lem}
\begin{proof}
Since $\Zs$ takes values in $\pi_\setS(R)$, it suffices to show that the 
restriction of $f$ to the domain $\pi_\setS(R)$ is a one-to-one function.
Assume, for the sake of contradiction, that there are two elements $a \neq b \in X_\setS$ 
such that $f(a) = f(b)$ and both $Pr(\Zs = a)$ and $Pr(\Zs = b)$ are non-zero.
Thus there must be ``preimages" of $a,b$ with respect to $\pi_\setS$; that is, elements $A, B \in R$ such that $\pi_\setS(A) = a$ and $\pi_\setS(B) = b$.
Furthermore, $A \neq B$ since $a \neq b$.

Without loss of generality, let $a <_{lex} b$, and define the vector $B'  \in \bXk$ by 
\[ 
B'_i = \left\{
\begin{array}{ll}
A_i ~(=a_i)& \text{for } i \in \setS \\ 
B_i & \text{for } i \notin \setS
\end{array} 
\right..
\]
For example, let $k=5$ and $\setS=\{2,4,5\}$, then setting 
$a=(a_2, a_4, a_5)$ and $b=(b_2, b_4, b_5)$, we have
\begin{align*}
B& \,=\,(B_1, b_2, B_3, b_4, b_5) \\
B'& \,=\,(B_1, a_2, B_3, a_4, a_5) .
\end{align*}

Note that since $f(a) = f(b)$ and $f$ is partition-determined with respect to $\setS$, it follows that $f(B) = f(B')$.
Due to the fact that $B$ was chosen to be the representative of $f(B)$, 
one must have $B <_{lex} B'$.
On the other hand, since $a <_{lex} b$, and since $B$ and $B'$ agree on all $i \notin \setS$,
we have that $B'<_{lex} B$.
This is a contradiction.
\end{proof}

The following result now becomes basic when considering such constructions.

\begin{thm} \label{thm:set-main}
Suppose $f$ is partition-determined with respect to $\collS$, and $\alpha$ is a 
fractional covering of $[k]$ using $\collS$. 
For any $Y \subseteq f(\bXk)$, we have that
\ben
|Y| \leq \prod_{\setS \in \collS} \left|f_\setS \left(f_{[k]}^{-1}(Y)\right)\right|^{\alpha_\setS}.
\een
\end{thm}
\begin{proof}
Let ${\bZ}$ be a random variable distributed uniformly on the elements of $R$ and let $Z_i = \pi_i({\bZ})$ for all $i \in [k]$.
Then by the usual fractional subadditivity of entropy, we have
\begin{equation} \label{entropy:1}
\log\big(|R|\big) = H({\bZ}) = H(Z_1, \dots, Z_k) \leq \sum_{\setS \in \collS} \alpha_\setS H(\Zs)
\end{equation}
where 
$\Zs = \pi_\setS({\bZ})$. 
By the chain rule for entropy, 
\begin{equation}\label{entropy:2}
H\big(\Zs \mid f(\Zs)\big) + H\big(f(\Zs)\big) = H\big(\Zs, f(\Zs)\big) = H\big(f(\Zs) \mid \Zs\big) + H(\Zs)
\end{equation}
for each $\setS \in \collS$.
Here, $H\big(f(\Zs) | \Zs\big) = 0$ since $f$ is completely determined by its input.
On the other hand, $H(\Zs | f(\Zs\big)) = 0$ by Lemma~\ref{entropy:L}, so Equation~\ref{entropy:2} reduces to $H(\Zs) = H\big(f(\Zs)\big)$.

Plugging this into Equation (\ref{entropy:1}) yields:
\[\log\big(|R|\big) \leq \sum_{\setS \in \collS} \alpha_\setS H\big(f(\Zs)\big) \leq \sum_{\setS \in \collS} \alpha_\setS \log \big( \left|f_\setS \left(R\right) \right|\big) \leq \sum_{\setS \in \collS} \alpha_\setS \log \Big( \big| f_\setS \big(f_{[k]}^{-1}(Y)\big) \big| \Big) \]
where the last inequality is due to the fact that $R \subseteq f_{[k]}^{-1}(Y)$.
Since $|Y|=|R|$, our claimed result is true.
\end{proof}


Considering the full compound set rather than a subset of it yields the following corollary. 

\begin{cor}\label{cor:cpd}
Suppose $f$ is partition-determined with respect to $\collS$, and $\alpha$ is a 
fractional covering of $[k]$ using $\collS$. Then
\ben
|f(\bXk)| \leq \prod_{\setS \in \collS} \left|f (\Xs)\right|^{\alpha_\setS}.
\een
\end{cor}

The only potential problem in obtaining Corollary~\ref{cor:cpd} from Theorem~\ref{thm:set-main}
is that $f^{-1}\naplus~f(\bX_{[k]})$ could in general be a superset of $\bX_{[k]}$, if the $X_i$ are finite
subsets of an ambient space $G$. However this is not a problem, since throughout we only require
$f$ to be defined on $\bX_{[k]}$ (and hence can work with the range of the inverse function being 
thought of as restricted to $\bX_{[k]}$).

Let us also mention here that Corollary~\ref{cor:cpd} is not a consequence of an underlying
submodularity property; in particular, one can generally find examples for which
\ben
\big|f_{\setS\cup\setT} \left(f_{[k]}^{-1}(Y)\right)\big| \cdot \big|f_{\setS\cap\setT} \left(f_{[k]}^{-1}(Y)\right)\big| 
> \big|f_{\setS} \left(f_{[k]}^{-1}(Y)\right)\big| \cdot \big|f_{\setT} \left(f_{[k]}^{-1}(Y)\right)\big| .
\een
Such counterexamples in the cases of $f$ corresponding to projections and sums
are discussed in Sections~\ref{ss:set-proj} and \ref{ss:set-abelian} respectively.
This phenomenon appears to be in stark contrast to the corresponding entropy statement (Theorem~\ref{thm:entSM}),
which asserts for strongly partition-determined functions such as sums
that such a submodularity holds for entropy. The discrepancy is explained by the fact that
one needs to invoke a set of representatives to pass from entropy to set cardinality.

\begin{rmk}
Let $Z_1, \ldots, Z_n$ be random variables supported by the sets $X_1, \ldots, X_n$ 
respectively, and note that for any subset $\setS\subset[n]$, the random variable 
$f_{\setS}(Z)$ is supported on the compound set $f_{\setS}(\bXn)$.
Thus the left hand side of Theorem~\ref{thm:entUB} has the bound
\be\label{ub-wanteq}
H(f_{[n]}(\bZ)) \leq \log |f_{[n]}(\bXn)| ,
\ee
while its right hand side has the bound
\ben
\sumS \as H\big( f_{\setS}(\bZ)\big) \leq \sumS \as \log \big| f_{\setS}(\bXn)\big| .
\een
Corollary~\ref{cor:cpd} says that these bounds themselves are ordered. This would be implied by 
Theorem~\ref{thm:entUB}  if we could find a product distribution on $(X_1,\ldots,X_n)$ that made their 
sum uniformly distributed on its range, since then \eqref{ub-wanteq} would simply hold with equality. 
Unfortunately, while it is always possible to find a joint distribution (with dependence) that makes the sum
uniformly distributed on its range, it is in general not possible to find such product distributions
(cf. \cite{MMT10:itw}).
Thus the hoped-for simple method of
proof of Corollary~\ref{cor:cpd} from Theorem~\ref{thm:entUB} as outlined above fails
because of the independence requirement of Theorem~\ref{thm:entUB}. 
The method of proof of Corollary~\ref{cor:cpd} in Section~\ref{ss:set-basic}
using the uniform distribution on a set of representatives is precisely designed
to address this problem.
\end{rmk}

\subsection{Corollaries on set projections}
\label{ss:set-proj}

There are some well-known inequalities for the cardinalities of set projections that come directly out of Theorem~\ref{thm:set-main}. 
We mention this way of proving them only
for completeness, and to illustrate the unifying nature of Theorem~\ref{thm:set-main}.

\begin{cor}
\label{cor:projset}
Let $\alpha$ be a fractional covering using the hypergraph $\collS$ on $[k]$.
Let $X_1, \ldots, X_k$ be arbitrary finite sets, and $Y\subset \bXk$.
Then
\ben
|Y| \leq \prodS |\pi_{\setS} (Y)|^{\as} .
\een
\end{cor}

\begin{proof}
Apply Theorem~\ref{thm:set-main} to the function $f_{\setS}=\pi_{\setS}$,
which is obviously partition-determined with respect to any collection $\collS$.
\end{proof}

For instance, Lemma~3.1 of Gyarmati, Matolcsi and Ruzsa \cite{GMR10}
follows from Corollary~\ref{cor:projset}. As suggested there, these inequalities
can be directly derived either from Shearer-type entropy inequalities (such
as those of \cite{MT10}) or from the so-called `Box Theorem' of Bollob\'as and Thomason \cite{BT95}. 

It is pertinent to note that the corresponding (stronger) submodularity inequality
\ben
|\pi_{\setS\cup\setT} (Y)| \cdot |\pi_{\setS\cap\setT} (Y)| \leq |\pi_{\setS} (Y)| \cdot |\pi_{\setT} (Y)|
\een 
does not hold. For a counterexample, consider $X_i=\{0,1\}$
for $i=1,2,3$, and let $Y=\{(0,0,0), (1,0,0), (0,1,0), (0,0,1), (1,0,1)\}$.
Then $|\pi_{\{1,2\}}(Y)|=3$, $|\pi_{\{2,3\}}(Y)|=3$, while $|\pi_{\{1,2,3\}}(Y)|=5$ and $|\pi_{\{2\}}(Y)|=2$.  
(This counterexample is implicit in \cite{BB07}.)

\subsection{Corollaries on abelian sumsets}
\label{ss:set-abelian}

We first write down  Corollary~\ref{cor:cpd} for the special case of the sum function.

\begin{cor}\label{cor:frac-sum-ineq}
Let $B_1, B_2, \dots, B_k \subset \calG$, where $(\calG,+)$ is an abelian group under the operation $+$. 
For any $\setS\subset [k]$, define
\be\label{bsp-defn}
B^{+}_{\setS}= \sum_{i\in\setS} B_i ,
\ee
which is well defined by commutativity of addition. Let $\alpha$ be any
fractional covering on $[k]$ using the collection $\collS$ of subsets of $[k]$. Then
\ben
|B^{+}_{[k]}|\leq \prodS |B^{+}_{\setS}|^{\as}.
\een
\end{cor}

For the case of $\collS_{k-1}$, this was independently proved in \cite[Theorem 1.2]{GMR10}, 
while for the case of regular hypergraphs $\collS$ (which can be shown to imply 
Corollary~\ref{cor:frac-sum-ineq}),
this was independently proved in \cite[Theorem 7]{BB07}. Observe that 
 the $\collS_{k-1}$ case implies the well known fact that the sequence $|lB|^{1/l}$ is non-increasing
 (here $lB$ as usual denotes the sum of $l$ copies of the set $B$).

\begin{rmk}\label{rmk:ruz-nonsm}
The ``fractional subadditivity'' of log cardinalities of sumsets captured by Corollary~\ref{cor:frac-sum-ineq}
leads one to wonder whether these log cardinalities might satisfy the stronger property (see, e.g., \cite{MT10}) of submodularity.
However, Ruzsa \cite[Part II, Theorem 1.9.5]{GR09:book} has observed that sumset cardinality is not log submodular.
Indeed, log submodularity of sumset cardinality would imply that $|kA|$
is a log concave function of $k$. Ruzsa showed that this is not the case, by proving that
if $|A|=n$ and  $|A+A|=m$, then $|A+A+A|$ can be anywhere between
$cm$ and $C \min(m^{3/2}, m^3/n^2)$, for some constants $c$ and $C$.
\end{rmk}

We now develop a more general variant of Corollary~\ref{cor:frac-sum-ineq} that allows for 
considering arbitrary subsets $D$ of the sumset $B^{+}_{[k]}$.  

\begin{thm}\label{thm:abelian}
Let $A, B_1, B_2, \dots, B_k \subset \calG$, where $(\calG,+)$ is an abelian group under the operation $+$. 
For any $\setS\subset [k]$, define $B^{+}_{\setS}$ as in \eqref{bsp-defn}. Let $\alpha$ be any
fractional covering on $[k]$ using the collection $\collS$ of subsets of $[k]$.
Then, for any $D \subseteq B^{+}_{[k]}$,
\ben
|A+D|^c \leq |D|^{c-1} \prodS |A+ B^{+}_{\setS}|^{\as} ,
\een
where $c=\sum_{\setS\in\collS} \as$.
\end{thm}

\begin{proof}
Set $X_i = B_{i}$ for $i\in[k]$, and $X_{k+1} = A$. Let ${\bf X}_{[k+1]}$ be the Cartesian product set.
Clearly, the collection of functions $f_{\setS}(x) = \sum_{i \in \setS} x_i$, for $\setS\subset [k+1]$,  is partition-determined 
with respect to any collection $\collS'$ of subsets of $[k+1]$.
Since $D \subseteq  f_{[k]}(X)$, set  $Q = f_{[k]}^{-1}(D)$, then $Q\subset \bXk$.
Note that $Y=D+A$ can be written as 
\ben
Y = \{ f(b_1, \dots, b_k,a):  (b_1, \dots, b_k) \in Q ,   a \in A \}  .
\een

Exactly as in the proof of Theorem~\ref{thm:ent-pl-r},  choose 
\ben 
\collS' = \left\{ [k] \right\} \cup \left\{ \setS': \setS'=\setS\cup\{k+1\} , \setS\in\collS \right\} ,
\een
and define the fractional covering $\gamma$ for $[k+1]$ by
\ben
\gamma_{\setS\cup\{k+1\}} = \frac{\alpha_{\setS}}{\sum_{\setT\in\collS} \alpha_{\setT}} ,
\een
for each set $\setS\in\collS$, and
\ben
\gamma_{[k]}= 1-\frac{1}{\sum_{\setT\in\collS} \alpha_{\setT}} .
\een
By Theorem~\ref{thm:set-main}, we have that
\ben 
|Y| \leq \prod_{\setS' \in \collS'} \left| f_{\setS'} \left( f_{[k+1]}^{-1}(Y) \right) \right|^{\gamma_{\setS'}} .
\een

Using the fact that $f_{[k+1]}^{-1}(Y)=Q\times A$, it is easy to see that
$f_{[k]} \left( f_{[k+1]}^{-1}(Y) \right) = D$, and that if $\setS'=\setS\cup\{k+1\}$
for $\setS\subset[k]$,
\ben
f_{\setS'} \left(f^{-1}(Y)\right) \subseteq B^{+}_{\setS} + A .
\een
The result follows.
\end{proof}

By applying Theorem~\ref{thm:abelian} to an $r$-regular hypergraph $\collS$,
for which $\as=\frac{1}{r}$ gives a fractional partition, we obtain the following corollary
(the Illustrative Set Result in the Introduction).

\begin{cor}\label{cor:abelian}
Under the assumptions of Theorem~\ref{thm:abelian}, if $\collS$ is an 
$r$-regular hypergraph, then for any $D \subseteq B^{+}_{[k]}$,
\ben
|A+D|^{|\collS|} \leq |D|^{{|\collS|}-r} \prodS |A+ B^{+}_{\setS}| .
\een
\end{cor}

Note that when $\collS=\collS_1$ is the collection of singleton sets, Corollary~\ref{cor:abelian}
reduces to 
\ben 
|A + D|^{k} \leq |D|^{k-1} \prod_{i = 1}^{k} |A + B_i| ,
\een
which was proved for $k=2$ by Gyarmati, Matolcsi, and Ruzsa \cite[Theorem~1.5]{GMR10}
and for general $k$ by Balister and Bollob\'as \cite[Theorem~8]{BB07}.
When $\collS=\collS_{k-1}$ is the collection of leave-one-out sets, Corollary~\ref{cor:abelian}
resolves a conjecture stated in \cite[Problem~1.6]{GMR10} and independently proved in \cite[Theorem~6.1]{GMR08}, namely if 
$\overline{B_i} = B_1 + \dots + B_{i-1} + B_{i+1} + \dots + B_k$ for $i = 1, \dots, k$,
\ben 
|A + D|^{k} \leq |D| \prod_{i = 1}^{k} |A + \overline{B_i}| .
\een
Clearly, various other choices of $\collS$ yield other similar corollaries. 

It is worth noting here that Theorem~\ref{thm:abelian} is precisely the set cardinality analogue
of inequality \eqref{ent-pl-r1}, which was part of Theorem~\ref{thm:ent-pl-r}.
We conjecture that the following set cardinality analogue of \eqref{ent-pl-r2} is also true.

\begin{conj}\label{conj:abelian}
Under the assumptions of Theorem~\ref{thm:abelian}, there exists
a nonempty set $A' \subseteq A$ such that
\ben
\frac{|A'+B^{+}_{[k]}|}{|A'|}  \leq  \prodS \bigg(\frac{|A+ B^{+}_{\setS}|}{|A|}\bigg)^{\as} .  
\een
\end{conj}

Analogous to inequality \eqref{ent-pl-r2b}, one can rewrite Conjecture~\ref{conj:abelian} as asserting
that there exists a nonempty set $A' \subseteq A$ such that
$
|A'+B^{+}_{[k]}|\leq (\prodS K_{\setS}^{\as}) |A'| 
$,
under the assumption that $|A+ B^{+}_{\setS}|\leq K_{\setS} |A|$ for each $\setS\in\collS$. For $r$-regular
hypergraphs, the conclusion reads $|A'+B^{+}_{[k]}|\leq (\prodS K_{\setS})^{1/r} |A'|$.
Clearly, when specialized to the collection $\collS_1$ of singletons,
this is just the original Pl\"unnecke-Ruzsa inequality for different summands \cite{Ruz89}.
It was proved for the collections $\collS_m$ consisting of all sets of size $m$ was
by Gyarmati, Matolcsi and Ruzsa \cite[Theorem 1.3]{GMR08}. 

\begin{rmk}
Interestingly, there is an essential difference between the set cardinality 
analogues of the inequalities \eqref{ent-pl-r1} and \eqref{ent-pl-r2}, 
although the entropy versions look very similar. Indeed, while Theorem~\ref{thm:abelian} 
holds for {\it any} $D \subseteq B^{+}_{[k]}$, it seems that one can only hope for 
Conjecture~\ref{conj:abelian} to be true for {\it some} $A'\subset A$  (ideally not too small) 
and $D=B^{+}_{[k]}$. Whether the latter difference is essential remains unclear to us, but the
reason that $A'$ cannot be taken to be $A$ is clear: it arises from the fact (discussed in 
Remark~\ref{rmk:ruz-nonsm}) that log cardinality of sumsets is {\it not} submodular
although it is fractionally subadditive, whereas the entropy of sums {\it is} submodular. 
To see the pertinence of this, observe that the $k=2$ case of the original Pl\"unnecke-Ruzsa 
inequality for different summands, written for $A'=A$, is just
\be\label{false}
|A+B_1+B_2| \cdot |A|  \leq |A+ B_1| \cdot |A+ B_2| \quad\quad \text{(false)}
\ee
i.e., nothing but log-submodularity of the
cardinality of sumsets. Furthermore, a careful look at the proofs of 
inequalities \eqref{ent-pl-r1} and \eqref{ent-pl-r2} reveal that only the latter used
the full strength of the submodularity of the entropy of sums. The usefulness of 
submodularity is also evident in the elegant and very recent work (released while this paper was being revised)
of Petridis \cite{Pet11}, whence it becomes clear that at the heart of the Pl\"unnecke-Ruzsa 
inequality is the fact that although \eqref{false} is not true, 
\ben
|A'+B_1+B_2| \cdot |A'|  \leq |A'+ B_1| \cdot |A'+ B_2|
\een
holds for some subset $A'$ of $A$.
\end{rmk}

\subsection{Corollaries on non-abelian sumsets}
\label{ss:set-nonabelian}

One particular generalization that has been explored recently is that to non-abelian groups.
In particular, given $X_1, \dots, X_k$ subsets of a non-abelian group $(\calG, \naplus)$, can we find similar bounds on 
$|X_1 \naplus \dots \naplus  X_k|$ as we did when the underlying group was abelian?
Unfortunately, the non-abelian addition function is no longer partition-determined with respect to any collection of subsets; 
however, it is partition-determined with respect to some collections of subsets, and so 
with a little added work we can still use Lemma~\ref{entropy:L} to find bounds.
To see that, in fact, the same bounds cannot hold, consider the following example.

\begin{example}\label{entropy:dihedral}
Let $\scrG = \{e, R, R^2, F, RF, R^2F \}$ be the dihedral group on $6$ elements, $S = \{e, F\}, T=\{R\}$, and $U = \{e, F\}$.
Then it is {\em not} the case that $|S \naplus T \naplus U|^2 \leq |S \naplus T||T \naplus U||S \naplus U|$.
\end{example}
\begin{proof}
On one hand, we have that $S \naplus T = \{R, FR\}, S \naplus U = \{e,F\}$, and $T \naplus U = \{R, RF\}$, and so $|S \naplus T||T \naplus U||S \naplus U|=8$.
On the other hand, $S \naplus T \naplus  U = \{ R, FR, RF, R^2 \}$ and so $|S \naplus T \naplus U|^2 = 16$.
\end{proof}

In fact, one can find examples (free groups) where the left hand side exceeds the right hand side by a power \cite{Ruz:comm}. 
The underlying reason that non-abelian groups cannot be bounded in such a way is that, 
as in the example above, $|S \naplus U|$ need not have any relation to $|S \naplus T \naplus U|$.
So any bound will need to find some way to link the two.
To do so, we will need to use stronger inequalities than the ones mentioned in the previous sections.
We will derive them using the simplified proof of Shearer's Lemma due to Llewellyn and Radhakrishnan 
(see \cite{Rad03}), and which has been used to find extended Shearer-type bounds 
(the most general of these appearing in \cite{Fri04, MT10}).
However, the lemma below seems to be disjoint from results mentioned previously 
in the literature (and is a considerable strengthening of the similar Han inequality).
The general idea is that the entropy of a collection of random variables can be compared 
to the sum of conditional entropies of pairs of random variables, each one conditioned on 
all of the random variables falling in between the pair.

\begin{lem}
\label{lem:induct}
Let $Z = Z_1, Z_2, \dots, Z_k$ be random variables, and define
\ben 
Z_{(i,j)} =  \{ Z_t : i < t < j \}
\een
for all $1 \leq i < j \leq k$.
Then, for $k \geq 2$,
\ben
(k-1)H(Z_1, Z_2, \dots, Z_k) \leq \sum_{i =1}^{k} \sum_{j > i}^k H(Z_i, Z_j \mid Z_{(i,j)})
\een
\end{lem}
\begin{proof}
We will prove this by induction on $k$.
The base case (when $k = 2$) is trivial, so assume the hypothesis to be true for $k-1$ random variables and consider a collection of $k$ random variables.
The general idea will be to peel off all of the pairs that contain the random variable $Z_k$ and then appeal to the induction hypothesis for the other $\binom{k-1}{2}$ pairs.
We write $H(Z)$ a total of $k-1$ times, in different forms:
\begin{align*}
H&(Z)=  & H&(Z_1, Z_k \mid Z_2,\dots, Z_{k-1})&+~H&(Z_2, \dots, Z_{k-1}) \\  
H&(Z)=H(Z_1 \mid Z_2, \dots, Z_{k-1}, Z_k) & +~H&(Z_2, Z_k \mid Z_3, \dots, Z_{k-1}) &+~H&(Z_3, \dots, Z_{k-1}) \\
H&(Z)=H(Z_1, Z_2 \mid Z_3, \dots, Z_{k-1}, Z_k) &+~H&(Z_3, Z_k \mid Z_4, \dots, Z_{k-1}) & +~H&(Z_4, \dots, Z_{k-1}) \\
&\vdots \qquad\qquad\qquad\qquad \vdots &&\qquad\quad\vdots&& \quad\vdots \\
H&(Z)=H(Z_1, \dots, Z_{k-3} \mid Z_{k-2}, Z_{k-1}, Z_k) &+~H&(Z_{k-2}, Z_k \mid Z_{k-1}) & +~H&(Z_{k-1}) \\
H&(Z)=H(Z_1, \dots, Z_{k-3}, Z_{k-2} \mid Z_{k-1}, Z_k) &+~H&(Z_{k-1}, Z_k)  &&
\end{align*}
Note that the middle terms are all of the type $H(Z_i, Z_k \mid Z_{(i,k)})$ (in particular, $H(Z_{k-1}, Z_k) = H(Z_{k-1}, Z_k \mid Z_{(k-1,k)})$ since $Z_{(k-1,k)}$ is empty).
Furthermore, removing a random variable from the conditioning can only increase the entropy, so if we remove $Z_k$ from all of the leftmost terms, we get 
\begin{align*}
H&(Z)~\leq & H&(Z_1, Z_k \mid Z_{(1,k)}) &+~~H&(Z_2, \dots, Z_{k-1}) \\  
H&(Z)~\leq~H(Z_1 \mid Z_2, \dots, Z_{k-1}) &+~~H&(Z_2, Z_k \mid Z_{(2,k)}) &+~~H&(Z_3, \dots, Z_{k-1}) \\
H&(Z)~\leq~H(Z_1, Z_2 \mid Z_3, \dots, Z_{k-1}) &+~~H&(Z_3, Z_k \mid Z_{(3,k)})& +~~H&(Z_4, \dots, Z_{k-1}) \\
&\vdots \qquad\qquad\qquad \vdots &&\qquad\quad\vdots && \quad\vdots \\
H&(Z)~\leq~H(Z_1, \dots, Z_{k-3} \mid Z_{k-2}, Z_{k-1}) &+~~H&(Z_{k-2}, Z_k \mid Z_{(k-2,k)}) & +~~H&(Z_{k-1}) \\
H&(Z)~\leq~H(Z_1, \dots, Z_{k-3}, Z_{k-2} \mid Z_{k-1}) &+~~H&(Z_{k-1}, Z_k \mid Z_{(k-1,k)}) &&
\end{align*}
And now if we sum all of these inequalities, the entries in the leftmost column can be combined with the entries in the rightmost column that is one level higher to get
\ben
(k-1)H(Z) \leq \sum_{i=1}^{k-1} H(Z_i, Z_k \mid Z_{(i,k)}) + (k-2)H(Z_1, Z_2, \dots, Z_{k-1})
\een
and the rest of the inequality follows from the induction hypothesis.
\end{proof}

Using Lemma~\ref{lem:induct}, we can give a collection of inequalities for non-abelian groups.
\begin{thm}
\label{thm:nonabelian}
Let $X_1, X_2, \dots, X_k$ be subsets of a non-abelian group, and define
\ben 
A(i,j) = \max \{ | X_i \naplus x_{i+1} \naplus \dots \naplus x_{j-1} \naplus X_j| : x_{i+1} \in X_{i+1}, \dots, x_{j-1} \in X_{j-1} \}
\een
for all $1 \leq i < j \leq k$.
Then, for $k \geq 2$,
\ben
|X_1 \naplus X_2 \naplus \dots \naplus X_k|^{k-1} \leq \prod_{1 \leq i < j \leq k} A(i,j) 
\een
\end{thm}
\begin{proof}
We define $Z$ as before.
By Lemma~\ref{lem:induct}, we have that 
\ben
(k-1)H(Z_1, Z_2, \dots, Z_k) \leq \sum_{i =1}^{k} \sum_{j > i}^k H(Z_i, Z_j \mid Z_{(i,j)})
\een
so it suffices to show that $H(Z_i, Z_j \mid Z_{(i,j)}) \leq \log{A(i,j)}$.
Note that in the abelian case, we would have been pleased enough with $H(Z_i, Z_j)$ but in the non-abelian case, the sumset $f$ that we used before is not (necessarily) partition-determined with respect to the partition $\{ \{i,j\}, \overline{\{i,j \}} \}$.
In light of this, for each subset $\setS \subset [k]$, we define the set function $g$ as 
\[
g_\setS(x_1, \dots, x_k) = \left \{
\begin{array}{rl}
x_{i_1} \naplus \dots \naplus x_{i_{|\setS|}} & \text{if $\setS$ forms a consecutive interval in $[k]$} \\
(y_1, y_2, \dots, y_k) & \text{where $y_i =x_i$ for $i \in \setS$ and $y_i = 0$ for $i \notin \setS$, otherwise}
\end{array}\right.
\]
It is easy to check that this $g$ is partition-determined with respect to any partition.
And while we have seemingly increased the sizes of a number of the sets that we wish to bound, the conditioning allows us to overcome this obstacle.
\ben 
H(Z_i, Z_j \mid Z_{(i,j)}) = H(Z_{[i,j]}) - H(Z_{(i,j)}) = H(g_{[i,j]}(Z)) - H(Z_{(i,j)}) = H(g_{[i,j]}(Z) \mid Z_{(i,j)})
\een
by Lemma~\ref{entropy:L}.
Furthermore, we have 
\begin{align*}
H(g_{[i,j]}(Z) \mid Z_{(i,j)}) 
&= \sum_{z_{(i,j)} \in Z_{(i,j)}} H(g_{[i,j]}(Z) \mid Z_{(i,j)} = z_{(i,j)})Pr(Z_{(i,j)}=z_{(i,j)}) \\
&\leq \max_{z_{(i,j)} \in Z_{(i,j)}} H(g_{[i,j]}(Z) \mid Z_{(i,j)} = z_{(i,j)}) \\
&\leq \log{ \max \{ | X_i \naplus x_{i+1} \naplus \dots \naplus x_{j-1} \naplus X_j| : x_{i+1} \in X_{i+1}, \dots, x_{j-1} \in X_{j-1} \} }\\
&= \log{A(i,j)}
\end{align*}
since, assuming that $Z_{(i,j)} = z_{(i,j)}$, the range of $g_{[i,j]}$ is at most $|X_i \naplus z_{i+1} \naplus \dots \naplus z_{j-1} \naplus X_j|$.
\end{proof}

Two interesting properties of our methods are revealed in Theorem~\ref{thm:nonabelian}.
The first is that $g$ does not have to (necessarily) be the obvious compound function --- in fact, the normal sumset 
function is not partition-determined in the non-abelian setting, and our choice of $g$ was able to overcome this.
The second is that the size of the sets of $g_{\overline{\setS}}$ are irrelevant in the bound as long as $\overline{\setS} \neq \setT$ for some $\setT$ that also creates
a partition. 
Thus if we are careful, we can possibly (artificially) inflate the size of many of the $\overline{\setS}$ sets to help make $g$ partition-determined without changing the bound.
The following corollary, which inspired Theorem~\ref{thm:nonabelian}, was originally proved by Ruzsa \cite{Ruz07}.

\begin{cor}
\label{cor:nonabelian}
Let $S,T,U$ be subsets of a non-abelian group.
Then
\ben
|S \naplus T \naplus U|^2  \leq \max_{t \in T}|S \naplus T| |T \naplus U||S \naplus t \naplus U|\,.
\een
\end{cor}

One could also attempt to use Corollary~\ref{cor:nonabelian} iteratively to obtain bounds in terms of the $A(i,j)$, however these will not match our Theorem~\ref{thm:nonabelian} (in particular, the variables will be necessarily asymmetric) \cite{Ruz:comm}.
Curiously, our methods seems to break down in other cases.  
In particular, it is unknown whether Theorem~\ref{thm:nonabelian} remains true in the asymmetric case 
(using the same covering set of pairs, but having different weights $\alpha_\setS$).
In addition, different covering sets are also still unexplored.
For example, the following problem posed by Ruzsa \cite{Ruz07} 
remains open.

\begin{conj}
Let $S,T,U,V$ be subsets of a non-abelian group.
Is it true that 
\ben
|S \naplus T \naplus U \naplus V|^3 \leq \max_{t,u}|S \naplus T \naplus U||S \naplus T \naplus u \naplus V||S \naplus t \naplus U \naplus V||T \naplus U \naplus V|?
\een
\end{conj}

Observe that the corresponding entropy inequality 
is {\it not} true. Indeed, 
if one chooses $Z_2 = Z_3$ to be non-deterministic and $Z_1 = Z_4 = 0$, then
the entropy inequality would claim that 
\begin{align*}
H(Z_2)&=H(Z_1,Z_2,Z_3,Z_4) \\
&\leq \frac{1}{3} \big[ H(Z_1,Z_2,Z_3) + H(Z_2,Z_3,Z_4) + H(Z_1,Z_3,Z_4 \mid Z_2) + H(Z_1,Z_2,Z_4 \mid Z_3)\big] \\
&= \frac{2}{3} H(Z_2) .
\end{align*}
This can only be true if $H(Z_2) \leq 0$ (which is false since $Z_2$ is non-deterministic).

\subsection{Corollaries on polynomial compound sets}
\label{ss:set-nonsum}

Although all of the examples and applications of partition-determined
set functions thus far have been associated with sumsets, 
this does not need to be the case.  Indeed, one can consider arbitrary ``compound sets'' obtained by plugging sets in
as arguments of any function involving well-defined operations on the ambient space. In other words,
for subsets $X_1,\ldots,X_k$ of some ambient space $\calX$, we will always use the notation
\ben
f(X_1,\ldots,X_k)=\{f(x_1,\ldots,x_k): x_1\in X_1,\ldots, x_k\in X_k  \}.
\een
When the ambient space is a group, the only operation available is the sum, and all compound sets are sumsets. 
When the ambient space is a ring, one may consider compound sets built from polynomials.
For particular ambient spaces, such as Euclidean space, the class of functions available
is extremely broad and therefore so is the class of compound sets that can be considered.

In this paper, we only illustrate the possible uses of our Theorem~\ref{thm:set-main} in the context of a ring.
The result below works for possibly non-commutative, non-unital rings. Note that while one must take 
care to maintain the order of terms within any finite monomial over a non-commutative ring,
one can still define polynomials (a sum of finite monomials) and corresponding polynomial functions 
although these polynomials are then no longer themselves commutative.
For our purposes, monomials in 2 indeterminates $\bx,\by$ over a non-commutative ring could include,
for instance, $a\bx^{b}c\by^{d}e\bx^{f}g$. We allow all such monomials with finitely many terms; all that matters is that
when two such monomials appear with opposite signs in a polynomial, they cancel to the additive identity 0
because of the commutativity of addition.

Given sets $A_1, \dots, A_k$, and a function $g$, we will write $g(A_1, \dots, A_k)$ to denote the set $\{ g(a_1, \dots, a_k) : a_i \in A_i \}$.
While it is perhaps overly pedantic to draw specific attention to notation, it is nonetheless 
necessary here, due to the proliferation of similar notation in the additive combinatorics community (see \cite{tvbook}).
In the usual notation for additive combinatorics, $AB +BA$ 
is used to refer to $\{ab + b'a' : a,a' \in A, b, b' \in B\}$.
However, when thinking of the compound set obtained by applying $g(a,b)=ab+ba$ to sets $A$ and $B$,
we need to consider the set $A\cdot B \oplus B \cdot A = \{ ab+ba : a \in A, b\in B \}$.
The symbols $\cdot$ and $\oplus$ are used to denote {\it bound} multiplication and addition operations,
i.e., when they appear in an expression, repeated appearances of a set symbol in that expression mean that
the {\it same} element of the set is substituted for each such appearance in computing the various instances of the expression.  
Our results could also be used to obtain results for the usual setting by using (for example) $A\cdot B \oplus C\cdot D$ 
where $A,D$ (likewise $B,C$) are a posteriori taken to be the same set, but they will tend to be 
weaker than the bounds that exploit the algebra of the underlying ring.

\begin{cor}\label{cor:cpd-ring}
Let $R$ be a (possibly non-commutative and non-unital) ring.
Suppose $f:R^m \ra R$ can be extended to a function $\bar{f}$ on $Q$
that is partition-determined with respect to $\collS$. For each $i\in [m]$, let 
$g_i\in R[\bx_1,\ldots, \bx_n]$ be a polynomial in $n$ indeterminates 
with coefficients in $R$, let $g_i:R^n\ra R$ be the corresponding polynomial function,
and let $g:R^n\ra R^m$ be the function whose $i$-th component is $g_i$. 
Let $F:R^n\ra R$ be the polynomial function associated with the polynomial 
$F\in R[\bx_1,\ldots, \bx_n]$ obtained by reducing the expression
\ben
f(\by_1,\ldots,\by_m) , \quad\text{with}\quad \by_i=g_i(\bx_1,\ldots,\bx_n)
\een
in the $\bx$-indeterminates, i.e., removing some pairs of monomials 
that are additive inverses of each other after substitution and expansion of terms. 
Then, for any collection $X_1,\ldots,X_n$ of finite subsets of $R$, and for any
fractional covering $\alpha$ with respect to the hypergraph $\collS$ on $[m]$,
\ben
|F(X_1,\ldots,X_n)| \leq \prodS |\bar{f}\naplus \pi_\setS \naplus g(X_1,\ldots,X_n)|^{\as} .
\een
\end{cor}

\begin{proof}
Observe that when $\by_i=g_i(\bx_1,\ldots,\bx_n)$ are substituted in $f(\by_1,\ldots,\by_m)$
and the resulting expression is expanded as a sum of monomials in the $\bx$-indeterminates, 
there is a polynomial $C\in R[\bx_1,\ldots, \bx_n]$ such that,
{\it before cancellation} of any monomials,
\ben
f(\by_1,\ldots,\by_m) =F(\bx_1,\ldots,\bx_n)+ C(\bx_1,\ldots,\bx_n)-C(\bx_1,\ldots,\bx_n) ,
\een
by the definition of $F$. Consequently, 
\ben
f(Y_1,\ldots,Y_m) =F(X_1,\ldots,X_n)\oplus C(X_1,\ldots,X_n) \ominus C(X_1,\ldots,X_n) .
\een
Since the set $C\ominus C$ is just the singleton $\{0\}$ containing the additive identity (because
these are bound operations), 
\begin{align*}
|F(X_1,\ldots,X_n)| 
&= |f(Y_1,\ldots,Y_m) | \\
&\leq \prodS |\bar{f}(Y_{\setS})|^{\as} \\
&= \prodS |\bar{f}\naplus \pi_\setS \naplus g(X_1,\ldots,X_n)|^{\as} ,
\end{align*}
where we used Corollary~\ref{cor:cpd} for the inequality.
\end{proof}

\begin{example}\label{ex:cpd1}
Let $A,B$ be subsets of a ring $R$. Then
\ben
|A^2\oplus B^2| \leq |(A\oplus B)^2| \cdot |A\cdot B\oplus B\cdot A| ,
\een
where the squares are understood to be bound squares.
\end{example}
\noindent{\it Remark:} When $R$ is the ring of real numbers $\RL$, the inequality in Example~\ref{ex:cpd1} implies that for any 
finite rectangular grid of points in the plane (corresponding to the Cartesian product $A\times B$), 
the number of distinct circles centered at the origin
passing through grid points is bounded by the product of the number of distinct lines of slope $-1$ passing
through grid points and the number of distinct hyperbolae with the axes as their asymptotes that pass
through the grid points.
\par\vspace{.0in}
\begin{proof}
Let $f(\by_1,\by_2)=\by_1^2-\by_2$. With $\by_1=\bx_1+\bx_2$ and $\by_2=\bx_1 \bx_2 +\bx_2 \bx_1$,
one finds $F(\bx_1,\bx_2)=(\bx_1+\bx_2)^2-\bx_1 \bx_2 -\bx_2 \bx_1= \bx_1^2 +\bx_2^2$.
Since $\bar{f}$ is partition-determined with respect to $\big\{\{1\},\{2\}\big\}$ if we set $\bar{f}(\by_1)=\by_1^2$
and $\bar{f}(\by_2)=\by_2$, one can apply Theorem~\ref{cor:cpd-ring}.
\end{proof}

A large class of further examples is provided by a general rule ---
whenever $F$ can be factorized, one can use the product function for $f$
(which is always partition-determined with respect to any $\collS$), and obtain
bounds for the cardinality of the polynomial set in terms of the cardinalities
of the ``factor sets''. 

\begin{cor}\label{lem:cpd-factor}
Let $R$ be a commutative ring. Suppose $F\in R[\bx_1,\dots, \bx_n]$ (written without any
redundant monomials) has a factorization of the form
\ben
F(\bx_1,\ldots, \bx_n)=\prod_{j\in [m]} g_j(\bx_1,\ldots, \bx_n) ,
\een
where each $g_j$ is also a polynomial (typically not involving some of the
$\bx$-indeterminates).
Then for any collection $X_1,\dots,X_n$ of finite subsets of $R$, and for any
fractional covering $\alpha$ with respect to the hypergraph $\collS$ on $[m]$,
\ben
|F(X_1,\ldots,X_n)| \leq \prodS \bigg|\prod_{j\in\setS} g_j(X_1,\ldots,X_n)\bigg|^{\as} .
\een
\end{cor}
\begin{proof}
Simply note that $F(\bx_1,\dots, \bx_n)=f(\by_1,\dots,\by_m)$, with the correspondence
$\by_j=g_j(\bx_1,\dots,\bx_n)$ and $f(\by_1,\dots,\by_m)=\prod_{j\in [m]} \by_j$.
Now $f$ has an obvious extension to $Q$, whose restriction to $Y_\setS$ is given by
$\bar{f}((\by_j:j\in\setS))=\prod_{j\in \setS} \by_j$. 
An application of Theorem~\ref{cor:cpd-ring} completes the proof.
\end{proof}

\section*{Acknowledgments}
We thank Imre Ruzsa for sharing the preprint \cite{GMR10}, for
helpful discussions, and for informing us of the independent and 
recent work of Balister--Bollob\'as. We also thank B\'ela Bollob\'as for 
promptly sending us the preprint \cite{BB07}, which contains 
results of independent interest.


\begin{thebibliography}{10}
\providecommand{\url}[1]{#1}
\csname url@rmstyle\endcsname
\providecommand{\newblock}{\relax}
\providecommand{\bibinfo}[2]{#2}
\providecommand\BIBentrySTDinterwordspacing{\spaceskip=0pt\relax}
\providecommand\BIBentryALTinterwordstretchfactor{4}
\providecommand\BIBentryALTinterwordspacing{\spaceskip=\fontdimen2\font plus
\BIBentryALTinterwordstretchfactor\fontdimen3\font minus
  \fontdimen4\font\relax}
\providecommand\BIBforeignlanguage[2]{{%
\expandafter\ifx\csname l@#1\endcsname\relax
\typeout{** WARNING: IEEEtran.bst: No hyphenation pattern has been}%
\typeout{** loaded for the language `#1'. Using the pattern for}%
\typeout{** the default language instead.}%
\else
\language=\csname l@#1\endcsname
\fi
#2}}


\bibitem{BB07}
P.~Balister and B.~Bollob\'as, ``Projections, entropy, and sumsets,''
{\tt arXiv:0711.1151v1 [math.CO]}, October 2007.

\bibitem{BM10:repi}
S.~Bobkov and M.~Madiman.
\newblock Reverse {B}runn-{M}inkowski and reverse entropy power inequalities for convex measures.
\newblock {\em J. Funct. Anal.}, to appear.

\bibitem{BM11:cras}
S.~Bobkov and M.~Madiman.
\newblock Dimensional behaviour of entropy and information.
\newblock {\em C. R. Acad. Sci. Paris S\'er. I Math.}, 349:201--204, F\'evrier
  2011.

\bibitem{BL91}
B.~Bollob{\'a}s and I.~Leader, ``Compressions and isoperimetric inequalities,''
  \emph{J. Combinatorial Theory Ser. A}, vol.~56, no.~1, pp. 47--62, 1991.

\bibitem{BT95}
B.~Bollob{\'a}s and A.~Thomason, ``Projections of bodies and hereditary
  properties of hypergraphs,'' \emph{Bull. London Math. Soc.}, vol.~27, no.~5,
  pp. 417--424, 1995.

\bibitem{CGFS86}
F.~Chung, R.~Graham, P.~Frankl, and J.~Shearer, ``Some intersection theorems
  for ordered sets and graphs,'' \emph{J. Combinatorial Theory, Ser. A},
  vol.~43, pp. 23--37, 1986.

\bibitem{CT91:book}
T.~Cover and J.~Thomas, \emph{Elements of Information Theory}.\hskip 1em plus
  0.5em minus 0.4em\relax New York: J. Wiley, 1991.

\bibitem{Fri04}
E.~Friedgut.
\newblock Hypergraphs, entropy, and inequalities.
\newblock {\em The American Mathematical Monthly}, 111(9):749--760, November
  2004.

\bibitem{Fuj78}
S.~Fujishige, ``Polymatroidal dependence structure of a set of random
  variables,'' \emph{Information and Control}, vol.~39, pp. 55--72, 1978.

\bibitem{GR09:book}
A.~Geroldinger and I.~Z. Ruzsa.
\newblock {\em Combinatorial number theory and additive group theory}.
\newblock Advanced Courses in Mathematics. CRM Barcelona. Birkh\"auser Verlag,
  Basel, 2009.
\newblock Courses and seminars from the DocCourse in Combinatorics and Geometry
  held in Barcelona, 2008.

\bibitem{GMR10}
K.~Gyarmati, M.~Matolcsi, and I.~Z. Ruzsa.
\newblock A superadditivity and submultiplicativity property for cardinalities
  of sumsets.
\newblock {\em Combinatorica}, 30(2):163--174, 2010.

\bibitem{GMR08}
K.~Gyarmati, M.~Matolcsi, and I.~Z. Ruzsa.
\newblock Pl\"unnecke's inequality for different summands.
\newblock In {\em Building bridges}, volume~19 of {\em Bolyai Soc. Math.
  Stud.}, pages 309--320. Springer, Berlin, 2008.

\bibitem{Han78}
T.~S. Han, ``Nonnegative entropy measures of multivariate symmetric
  correlations,'' \emph{Information and Control}, vol.~36, no.~2, pp. 133--156,
  1978.

\bibitem{KV83}
V.~A. Ka{\u\i}manovich and A.~M. Vershik.
\newblock Random walks on discrete groups: boundary and entropy.
\newblock {\em Ann. Probab.}, 11(3):457--490, 1983.

\bibitem{Mad08:itw}
M.~Madiman.
\newblock On the entropy of sums.
\newblock In {\em Proc. IEEE Inform. Theory Workshop}. Porto, Portugal, 2008.

\bibitem{Mad11:laa}
M.~Madiman, ``Determinant and trace inequalities for sums of positive-definite
  matrices,'' \emph{Preprint}.

\bibitem{MB07}
 M.~Madiman and A.~Barron, ``Generalized entropy power inequalities and
  monotonicity properties of information,'' \emph{IEEE Trans. Inform. Theory},
  vol.~53, no.~7, pp. 2317--2329, July 2007.

\bibitem{MMT10:itw}
M.~Madiman, A.~Marcus, and P.~Tetali.
\newblock Information-theoretic inequalities in additive combinatorics.
\newblock In {\em Proc. IEEE Inform. Theory Workshop}, Cairo, Egypt, January
  2010.

\bibitem{MT07:isit}
M.~Madiman and P.~Tetali.
\newblock Sandwich bounds for joint entropy.
\newblock {\em Proceedings of the IEEE International Symposium on Information Theory, Nice}, June 2007.

\bibitem{MK10:isit}
M.~Madiman and I.~Kontoyiannis.
\newblock The entropies of the sum and the difference of two {IID} random
  variables are not too different.
\newblock In {\em Proc. IEEE Intl. Symp. Inform. Theory}, Austin, Texas, June
  2010.

\bibitem{MT10}
M.~Madiman and P.~Tetali.
\newblock Information inequalities for joint distributions, with
  interpretations and applications.
\newblock {\em IEEE Trans. Inform. Theory}, 56(6):2699--2713, June 2010.

\bibitem{Mat07}
F.~Mat\'{u}s, ``Two constructions on limits of entropy
  functions,'' \emph{IEEE Trans. Inform. Theory}, vol.~53, no.~1, pp. 320--330,
  2007.

\bibitem{Nat96:book}
M.~Nathanson, \emph{Additive Number Theory: Inverse Problems and Geometry of
  Sumsets}, ser. Graduate Texts in Mathematics.\hskip 1em plus 0.5em minus
  0.4em\relax Springer, 1996, no. 165.

\bibitem{Pet11}
G.~Petridis, ``New Proofs of Pl\"unnecke-type Estimates for Product
Sets in Groups,'' \emph{Preprint}. Available:
  \url{http://arxiv.org/abs/1101.3507v2} 

\bibitem{Plu70}
H.~Pl{\"u}nnecke.
\newblock Eine zahlentheoretische {A}nwendung der {G}raphentheorie.
\newblock {\em J. Reine Angew. Math.}, 243:171--183, 1970.

\bibitem{Rad03}
J. Radhakrishnan.
\newblock Entropy and counting.
\newblock In {\em Computational Mathematics, Modelling and Algorithms} (ed. J. C. Misra), Narosa, 2003.
 [Online]. Available:
  \url{http://www.tcs.tifr.res.in/~jaikumar/Papers/EntropyAndCounting.pdf} 

\bibitem{Ruz:comm}
I.~Z. Ruzsa, personal communication.

\bibitem{Ruz89}
I.~Z. Ruzsa.
\newblock An application of graph theory to additive number theory.
\newblock {\em Scientia Ser. A Math. Sci. (N.S.)}, 3:97--109, 1989.

\bibitem{Ruz90}
I.~Z. Ruzsa.
\newblock Addendum to: An application of graph theory to additive number
  theory.
\newblock {\em Scientia Ser. A Math. Sci. (N.S.)}, 4:93--94, 1990/91.

\bibitem{Ruz07}
I.~Z. Ruzsa, ``Cardinality questions about sumsets,'' in \emph{Additive
  combinatorics}, ser. CRM Proc. Lecture Notes.\hskip 1em plus 0.5em minus
  0.4em\relax Providence, RI: Amer. Math. Soc., 2007, vol.~43, pp. 195--205.

\bibitem{Ruz09}
I. Z. Ruzsa, ``Sumsets and entropy,'' Random Struct. Algorithms, 34, no. 1, 2009, pp. 1--10.

\bibitem{Tao09:blog}
T.~Tao.
\newblock An entropy Pl\"unnecke-Ruzsa inequality.
\url{http://terrytao.wordpress.com/2009/10/27/an-entropy-plunnecke-ruzsa-inequality/}.

\bibitem{Tao10}
T.~Tao.
\newblock Sumset and inverse sumset theory for {S}hannon entropy.
\newblock {\em Combin. Probab. Comput.}, 19(4):603--639, 2010.

\bibitem{TV06:unpub}
T.~Tao, and V.~Vu, ``Entropy methods,'' \emph{Unpublished note},
  2006.  [Online]. Available:
  \url{http://www.math.ucla.edu/~tao/preprints/Expository/chapter_entropy.dvi}

\bibitem{tvbook}
T.~Tao and V.~Vu, { \em Additive Combinatorics}, Cambridge Studies in Advanced Mathematics 105, CUP 2006.

\end{thebibliography}
\end{document}